\documentclass[10pt]{article}
\usepackage{amsfonts,color}
\usepackage{amsmath,amssymb}
\usepackage{graphicx}
\usepackage{footnote}
\usepackage[utf8]{inputenc}	
\usepackage{mathtools}
\usepackage{amsthm,wasysym}

\usepackage[english]{babel}

\usepackage{bbm}

\newcommand{\id}{{\boldsymbol{\mathbbm{1}}}}

\newcommand{\Chi}{\raisebox{0.5ex}{\mbox{{\Large $\chi$}}}}
\newcommand{\Partial}{\raisebox{0ex}{\mbox{{\large $\partial$}}}}
\allowdisplaybreaks[1]

\newcommand{\ChiCaption}{\mbox{{\Large $\chi$}}}
\newcommand{\PartialCaption}{\mbox{{\large $\partial$}}}
\makeindex

\newtheorem{theorem}{Theorem}[section]
\newtheorem{lemma}[theorem]{Lemma}
\newtheorem{remark}[theorem]{Remark}
\newtheorem{proposition}[theorem]{Proposition}

\newtheorem{definition}[theorem]{Definition}
\newcommand{\R}{\mathbb{R}}

\DeclareMathOperator{\diag}{diag}
\DeclareMathOperator{\Sym}{Sym}

\DeclareMathOperator{\dev}{dev}

\def\barr{\begin{array}}
	
	\def\tr{\textrm{tr}}

	\def\dd{\displaystyle}

\begin{document}
	\title{Comparison of   isotropic elasto-plastic models  for the  plastic metric tensor\, $C_p=F_p^T\, F_p$}
	\author{
		Patrizio Neff\thanks{Corresponding author: Patrizio Neff,  \ \ Head of Lehrstuhl f\"{u}r Nichtlineare Analysis und Modellierung, Fakult\"{a}t f\"{u}r
			Mathematik, Universit\"{a}t Duisburg-Essen,  Thea-Leymann Str. 9, 45127 Essen, Germany, email: patrizio.neff@uni-due.de}\quad
		and \quad
		Ionel-Dumitrel Ghiba\thanks{Ionel-Dumitrel Ghiba, \ \ \ \ Lehrstuhl f\"{u}r Nichtlineare Analysis und Modellierung, Fakult\"{a}t f\"{u}r Mathematik,
			Universit\"{a}t Duisburg-Essen, Thea-Leymann Str. 9, 45127 Essen, Germany;  Alexandru Ioan Cuza University of Ia\c si, Department of Mathematics,  Blvd.
			Carol I, no. 11, 700506 Ia\c si,
			Romania; and  Octav Mayer Institute of Mathematics of the
			Romanian Academy, Ia\c si Branch,  700505 Ia\c si, email: dumitrel.ghiba@uni-due.de, dumitrel.ghiba@uaic.ro}}
	
	\maketitle
	
	\begin{center}
		\textit{Dedicated to Michael Ortiz on the occasion of his 60th birthday with great admiration}
	\end{center}

	\begin{abstract}
	\noindent We discuss  in detail existing  isotropic elasto-plastic models  based on  6-dimensional flow rules for the positive definite plastic metric tensor $C_p=F_p^T\, F_p$  and highlight their properties and interconnections. We show that seemingly different models are equivalent in the isotropic case.
	\end{abstract}

	\section{Introduction}

	Since the early days of the introduction of the multiplicative decomposition into computational elasto-plasticity, the need was felt to reduce the level of complexity and to discard the concept of a plastic rotation in the completely isotropic setting. This means to consider a flow rule not for the {\bf plastic distortion} $F_p$ (9-dimensional) \cite{NeffGhibaPlasticity,steigmann2011mechanically,gupta2011aspects,Reese97a,dettmer2004theoretical,Simo85,ortiz1986analysis,cuitino1992material}, but to consider directly a flow rule for the {\bf plastic metric tensor}  $C_p=F_p^TF_p\in{\rm PSym}(3)$ (6-dimensional) \cite{shutov2013analysis,GrandiStefanelli,shutov2008finite,brepols2014numerical,vladimirov2008modelling}, which is then automatically invariant under left-multiplication of $F_p$ with a plastic rotation. The plastic distortion is in general incompatible $F_p\neq \nabla \psi_p$, as is the plastic metric \break $C_p\neq \nabla\psi_p^T\nabla\psi_p$. A formulation in the plastic metric $C_p$ is particular attractive because it circumvents problems associated with the intermediate configuration introduced by the multiplicative decomposition, which is trivially non-unique since
	\begin{align*}
		F=F_e\cdot F_p=F_e\cdot Q^T\cdot Q\cdot F^p=F_e^*\cdot F_p^*, \quad Q\in{\rm SO}(3).
	\end{align*}
	Several proposals with the aim of removing the non-uniqueness of the intermediate configuration have been given in the literature. Our comparative study is related to the following models:
	Simo's model \cite{simo1993recent} (Reese and Wriggers \cite{Reese97a}, Miehe \cite{Miehe92}); Miehe's model \cite{Miehe95}; Lion's model \cite{lion1997physically} (Helm \cite{helm2001formgedachtnislegierungen}), Dettmer-Reese \cite{dettmer2004theoretical});
	Simo and Hughes' model   \cite{Simo98b};
	Helm's model \cite{helm2001formgedachtnislegierungen} (Vladimirov, Pietryga and Reese \cite{vladimirov2008modelling}, Shutov and  Krei{\ss}ig \cite{shutov2008finite},  Reese and Christ  \cite{reese2008finite}, Brepols, Vladimirov and  Reese \cite{brepols2014numerical}, Shutov and Ihlemann \cite{shutov2013analysis});
	Grandi and Stefanelli's model \cite{GrandiStefanelli} (Frigeri and Stefanelli \cite{frigeri2012existence}).
	All these models are given with respect to different configurations, either the reference configuration, the intermediate configuration or the current configuration. In order to be able to compare them, it is necessary to transform all to the same configuration for that purpose. In our case we choose the reference configuration. Moreover, any explicit dependence on $F_p$ instead of $C_p$ in the model formulation must be able to be subsumed  into a dependence on $C_p$ alone in the isotropic case. A major body of our work consists in showing this for the models under consideration.

	The paper  is structured as follows. After a paragraph giving some definitions which generalize the concepts from small strain-additive plasticity to finite strain plasticity we established some auxiliary results. Then we discuss  existing 6-dimensional flow rules from the literature. The main properties of the investigated isotropic plasticity models are summarized in Figure \ref{plastmodeldiagram} and Figure \ref{plastmodeldiagram2}. Finally, in the appendix, we obtain explicit formulas for some of the isotropic plasticity models.

	\subsection{Consistent isotropic finite plasticity model for the plastic metric tensor $C_p$}
	
	In this paper, we use the standard Euclidean scalar product on $\R^{3\times 3}$  given by
	$\langle {X},{Y}\rangle:=\tr{(X Y^T)}$, and thus the Frobenius tensor norm is
	$\|{X}\|^2=\langle {X},{X}\rangle$. The identity tensor on $\R^{3\times 3}$ will be denoted by $\id$, so that
	$\tr{(X)}=\langle {X},{\id}\rangle$. We let $\Sym(3)$ and $\rm PSym(3)$ denote the symmetric and positive definite symmetric tensors respectively. We adopt
	the usual abbreviations of Lie-group theory. {Here and i}n the following the superscript
	$^T$ is used to denote transposition,  ${\rm sym}\, X=\frac{1}{2}(X+X^T)$ denotes the symmetric part of the matrix $X\in \R^{3\times 3}$, while $\dev_3 X=X-\frac{1}{3}\, \tr(X)\cdot \id$ represents the deviatoric (trace free) part  of the matrix $X$.

	The classical concept of {\bf associated perfect plasticity} is uniquely defined in the case of small strain-additive  plasticity. In this case the  total symmetric  strain is decomposed additively into elastic and plastic parts $\varepsilon=\varepsilon_e+\varepsilon_p$ and the rate-independent evolution law for the symmetric plastic strain $\varepsilon_p$ is given in subdifferential format
	\begin{align*}
		\frac{\rm d}{\rm dt}[\varepsilon_p]\in  \Partial {\Chi}(\Sigma_{\rm lin}),\qquad \tr(\varepsilon_p)=0,
	\end{align*}
	where  $\Partial \Chi$ is the subdifferential of the indicator function $\Chi$ of the convex elastic domain $$\mathcal{E}_{\rm e}({\Sigma_{\rm lin}},\frac{2}{3}\, {\boldsymbol{\sigma}}_{\!\mathbf{y}}^2)=\left\{{\Sigma_{\rm lin}}\in{\rm Sym}(3) \big|\,\ \|\dev_3
	{\Sigma_{\rm lin}}\|^2\leq\frac{2}{3}\, {\boldsymbol{\sigma}}_{\!\mathbf{y}}^2\right\}\subset  {\rm Sym}(3)$$  and  $\Sigma_{\rm lin}:=-D_{\varepsilon_p}[W_{\rm lin}(\varepsilon-\varepsilon_p)]$ is the thermodynamic driving stress  of the plastic process. Here, $\Sigma_{\rm lin}$ is clearly symmetric.
	
	In such a way, the principle of {\bf maximum dissipation} (equivalent to the convexity of the elastic domain and normality of the flow direction) is satisfied. The structure of associated flow rules in geometrically nonlinear theories is by far not as trivial as in the geometrically linear models. However, in this work we use:
	\begin{definition}{\rm (geometrically nonlinear associated plastic flow)} \label{definitionpld}
		We call a plastic flow rule for some plastic variable $P$ (whether symmetric or not) associated, whenever the flow rule can be written as
		\begin{equation*}
			\frac{\rm d}{\rm dt}[P]\, P^{-1}\in  \Partial {\Chi}(\Sigma)\qquad \text{or}\qquad \sqrt{P}\frac{\rm d}{\rm dt}[P^{-1}]\, \sqrt{P}\in  f=\Partial {\Chi}(\Sigma) ,
		\end{equation*} where  $\Sigma$ is some symmetric or non-symmetric stress tensor.
		Here, $\frac{\rm d}{\rm dt}[P^{-1}]\, P$ is the correct format for the time derivative (it will lead to an exponential update, see the implicit method based on the exponential mapping considered in \cite{shutov2013explicit}). Moreover, we require that $\Chi$ is the indicator function of some {\bf convex} domain in the $\Sigma$-stress space.
	\end{definition}
	
	After liniarization (small strain-additive approximation) this condition is equivalent to classical associated plasticity.
	Further, let us also remark that a metric is by definition symmetric and positive definite, i.e. $C_p\in{\rm PSym}(3)$.
	\begin{definition}{\rm (consistent isotropic finite plasticity model for plastic metric tensor $C_p$)}\label{consistentdef}
		We say that an associated plastic flow  rule, in the sense of Definition \ref{definitionpld},  for the plastic metric  tensor $C_p$ is consistent, whenever:
		\begin{itemize}
			\item[i)] it is thermodynamically  correct, i.e. the reduced dissipation inequality is satisfied;
			\item[ii)] plastic incompressibility: the constraint $\det C_p(t)=1$ for all $t\geq 0$ follows from the flow rule;
			\item[iii)]  $\,C_p(t)\in{\rm PSym}(3)$ for all $t>0$ if $C_p(0)\in{\rm PSym}(3)$.
		\end{itemize}
	\end{definition}
	As we will see from the next Lemma \ref{lemmaPSym}, our requirement iii) follows if $C_p(t)\in{\rm Sym}(3)$ for all $t\geq 0$, $C_p(0)\in{\rm PSym}(3)$ and if ii) is satisfied.
	\newpage
	We finish our setup of preliminaries  with the following definitions:
	
	\begin{definition}{\rm (reduced dissipation inequality-thermodynamic consistency)} For a given energy $W$, we say that the reduced dissipation inequality along the plastic evolution is satisfied if and only if
		\begin{align*}
			\frac{\rm d}{\rm dt}[W(F\,F_p^{-1}(t)]=\frac{\rm d}{\rm dt}[\widetilde{W}(C\,C_p^{-1}(t)]=\frac{\rm d}{\rm dt}[\Psi(C,C_p(t)]\leq 0
		\end{align*}
		for all constant in time $F$ (viz. $C=F^TF$), depending in which format the elastic energy is given.
	\end{definition}
	
	\begin{definition} {\rm (Loss of ellipticity in the elastic domain)} We say that the elasto-plastic formulation preserves ellipticity in the elastic domain whenever
		the purely elastic response in elastic unloading of the material remains rank-one convex for arbitrary large given plastic pre-distortion, see \cite{NeffGhibaAdd,GhibaNeffMartin}.
	\end{definition}

	\subsection{Auxiliary results}

	We consider   the multiplicative decomposition of the deformation gradient \cite{kroner1955fundamentale,kroner1958kontinuumstheorie,kroner1959allgemeine,lee1969elastic,neff2009notes,Neff_Knees06}
	and we define, accordingly, the elastic and plastic  strain tensors
	\begin{align}
		C_e&:=F_e^TF_e\in{\rm PSym}(3),\qquad B_e:=F_e\,F_e^T\in{\rm PSym}(3),\notag\\ C_p&:=F_p^TF_p\in{\rm PSym}(3).\notag
	\end{align}
	Let us also define the  stress tensors
	\begin{align}
		\Sigma:&=2\,C\, D_C[\widehat{W}(C)]=2\, D_{\log C}[\overline{W}(\log C)]=D_{\log U}[\check{W}(\log U)]\notag\\
		&=U\, D_U[W(U)]=F^TD_F[W(F)]\,,\notag\\
		\tau:&=2\, D_B[\widehat{W}(B)]\, B=2\, D_{\log B}[\overline{W}(\log B)]=D_{\log V}[\check{W}(\log V)]\notag\\
		&=V\, D_V[W(V)]=2\, F\,D_C[\widehat{W}(C)]\,F^T.\notag
	\end{align}
	The tensor $\Sigma=C\cdot S_2(C)$, where  $S_2=2\,D_C[W(C)]$ is
	the second  Piola-Kirchhoff stress tensor, is sometimes called the \textbf{Mandel stress tensor}  and it holds $\dev_3 \Sigma_{e}=\dev_3 \Sigma_{\rm E}$, where  $\Sigma_{\rm E}$ is the elastic
	\textbf{Eshelby tensor}
	$$
	\Sigma_{\rm E}:=F_e^TD_{F_e}[W({F_e})]-W(F_e)\cdot \id=D_{\log C_e}[\overline{W}(\log C_e)]-\overline{W}(\log C_e)\cdot \id,
	$$
	driving the plastic evolution (see e.g. \cite{neff2009notes,maugin1994eshelby,cleja2000eshelby,cleja2003consequences,cleja2013orientational}), while $\tau$ is the {\bf Kirchhoff stress tensor} and $\Sigma_{e}$ is defined in Remark \ref{plastrem1}.

	\begin{remark}\label{plastrem1}
		We also need to consider the following elasto-plastic stress tensors:
		\begin{align}
			\Sigma_{e}:&=2\,C_e\, D_{C_e}[\widehat{W}({C_e})]=2\, D_{\log {C_e}}[\overline{W}(\log {C_e})]=D_{\log U_e}[\check{W}(\log U_e)]\notag\\
			&=U_e\,
			D_{U_e}[W(U_e)]=F_e^TD_{F_e}[W(F_e)]\,,\notag\\
			\tau_e:&=2\, D_{B_e}[\widehat{W}({B_e})]\, {B_e}=2\, D_{\log {B_e}}[\overline{W}(\log {B_e})]=D_{\log V_e}[\check{W}(\log V_e)]\notag\\=
			&V_e\, D_{V_e}[W(V_e)]=2\,
			{F_e}\,D_{C_e}[\widehat{W}({C_e})]\,{F_e}^T.\notag
		\end{align}
		
		The following  relation holds true:
		\begin{align}\label{7.8}
			\Sigma=F^T\tau \, F^{-T}, \qquad \Sigma_{e}=F_e^T\tau_e \, F_e^{-T}.
		\end{align}
		Note that \eqref{7.8}   is not at variance with symmetry of  $\Sigma$  and $\Sigma_e$ in case of isotropy.
	\end{remark}
	Using the fact that for given $F_e\in {\rm GL}^+(3)$ it holds $ \|F_e^TSF_e^{-T}\|^2\geq \frac{1}{2}\|S\|^2$ for all  $S\in {\rm Sym}(3)$, the constant being independent of $F_e$ \cite{NeffGhibaPlasticity},  we obtain the estimate
	\begin{equation*}
		\|\dev_3 \Sigma_{e}\|=\|F_e^T(\dev_3\tau_e)F_e^{-T}\|\geq \frac{1}{\sqrt{2}}\|\dev_3\tau_e\|,
	\end{equation*}
	which    is valid for general anisotropic materials. Indeed, since
	\begin{align*}
		\dev_3 \Sigma_{e}=\dev_3(F_e^T \tau_eF_e^{-T})&=F_e^T \tau_eF_e^{-T}-\frac{1}{3}{\rm tr}(F_e^T \tau_eF_e^{-T})\cdot \id\notag\\
		&=
		F_e^T (\tau_e-\frac{1}{3}{\rm tr} (\tau_e))\cdot \id)F_e^{-T},
	\end{align*}
	we  have
	\begin{align*}
		\dev_3 \Sigma_{e}=F_e^T (\dev_3\tau_e)F_e^{-T}, \qquad \dev_3 \tau_e=F_e^{-T} (\dev_3\Sigma_{e})F_e^{T}, \qquad {\rm tr}(\Sigma_{e})={\rm tr}(\tau_e).
	\end{align*}
	However, $\|\dev_3 \Sigma_{e}\|\neq \|\dev_3 \tau_{e}\|$ for general anisotropic materials.  Let us  remark that for elastically isotropic materials we
	have from the representation formula for isotropic tensor functions
	\begin{align}\label{isoplastalpha}
		D_{C_e}[\widehat{W}({C_e})]&=\alpha_1\, \id +\alpha_2\, C_e+\alpha_3 \, C_e^2\in{\rm Sym}(3),\notag\\
		\Sigma_{e}=2\,C_e\cdot D_{C_e}[\widehat{W}({C_e})]&=2\, C_e\, (\alpha_1\, \id +\alpha_2\, C_e+\alpha_3 \, C_e^2)\in{\rm Sym}(3),
	\end{align}
	where
	\begin{align*}
		\alpha_1&=\frac{2}{I_3^{1/2}({C_e})}\left[I_2({C_e})\,\frac{\partial W}{\partial I_2({C_e})}+I_3({C_e})\frac{\partial W}{\partial I_3({C_e})}\right], \quad \alpha_2=\frac{2}{I_3^{1/2}({C_e})}\frac{\partial W}{\partial I_1({C_e})},\notag\\
		\alpha_3&=-{2\,I_3^{1/2}({C_e})}\frac{\partial W}{\partial I_2({C_e})}
	\end{align*}
	are scalar functions of the invariants of $C_e$, which are functions of $C\, C_p^{-1}$, see Lemma \ref{lemaplasticn}. This leads us to
	
	\begin{lemma}\label{remarktauesigmae}For the isotropic case $\|\dev_3 \Sigma_e\|=\|\dev_3\tau_e\|$.
	\end{lemma}
	\begin{proof}  For the isotropic case we have $\tau_e\, B_e=B_e \, \tau_e$, which implies
		\begin{align*}
			\|\dev_3 \Sigma_e\|^2&
			=\langle F_e^T\, (\dev_3\tau_e) \, F_e^{-T},F_e^T\, (\dev_3\tau_e) \, F_e^{-T}\rangle=\langle B_e\, (\dev_
			3\tau_e) , (\dev_3\tau_e) \, B_e^{-1}\rangle\notag\\
			&=\|\dev_3\tau_e\|^2,
		\end{align*}
		and the proof is complete.\qedhere
	\end{proof}
	We also  consider the following tensor
	\begin{align}\label{definitiesigma}
		\widetilde{\Sigma}:=2\,C\, D_C[\widetilde{W}(C\, C_p^{-1})]=2\,C\, D[\widetilde{W}(C\, C_p^{-1})]\, C_p^{-1}\not\in {\rm Sym}(3),\end{align}
	which is not symmetric, in general. For instance, for  the simplest Neo-Hooke energy $W(F_e)=\tr(C_e)=\tr(C\, C_p^{-1})$ we have   $D \widetilde{W}(C\, C_p^{-1})=\id$ and $\widetilde{\Sigma}=2\,C\, C_p^{-1}\not\in {\rm Sym}(3)$.

	\begin{lemma}\label{lemaplasticn}
		Any isotropic and objective free energy  $W$ defined in terms of $F_e$ can be expressed as
		\begin{align}\label{fecp}
			W(F_e)=\widetilde{W}(C\,C_p^{-1})=\widetilde{W}(F^T F(F_p^TF_p)^{-1}).
		\end{align}
	\end{lemma}
	\begin{proof}
		It is clear that any objective elastic  energy $W(F_e)$ which is isotropic w.r.t. $F_e$, can be expressed in terms of the invariants of $C_e$, i.e.
		\begin{align*}
			W(F_e)&=\Psi(I_1(C_e),I_2(C_e),I_3(C_e)),\\
			I_1(C_e)&=\tr(C_e)=\tr(B_e),\quad
			I_2(C_e)=\tr({\rm Cof}\, C_e)=\tr({\rm Cof}\, B_e),\notag\\
			I_3(C_e)&=\det C_e=\det B_e.\notag
		\end{align*}
		
		Now every invariant can be rewritten as follows
		\begin{align}\label{7.5}
			I_1(C_e)&=\langle C_e,\id\rangle=\langle F_e^TF_e, \id\rangle
			=\langle  F_p^{-T}F^T\,(F\, F_p^{-1}), \id\rangle=\langle  C, C_p^{-1}\rangle
			=\tr(C\,C_p^{-1})\notag\\
			&=I_1(C\,C_p^{-1}),\\
			I_2(C_e)&=\langle {\rm Cof}\, C_e,\id\rangle=\det C_e\,\langle  C_e^{-T},\id\rangle=
			\det (F_p^{-T}C\, F_p^{-1})\,\langle  [F_p^{-T}C\, F_p^{-1}]^{-T},\id\rangle\notag\\
			&
			=\det C\det C_p^{-1}\,\langle  C^{-T}, F_p^{T}\, F_p\rangle\notag=
			\det (C\, C_p^{-1})\,\langle  C^{-T}C_p^T, \id\rangle\notag
			=
			\tr({\rm Cof}(C\,C_p^{-1}))\notag\\
			&=I_2(C\,C_p^{-1}),\notag\\
			I_3(C_e)&=\det\, C_e=
			\det (F_p^{-T}C\, F_p^{-1})
			=\det\, C\,\det\, C_p^{-1}=I_3(C\,C_p^{-1}).\notag
		\end{align}
		Therefore, we obtain
		\begin{align*}
			W(F_e)&=\Psi(I_1(C_e),I_2(C_e),I_3(C_e))\notag\\&=\Psi(I_1(C\,C_p^{-1}),I_2(C\,C_p^{-1}),I_3(C\,C_p^{-1}))=\widetilde{W}(C\,C_p^{-1}),
		\end{align*}
		and the proof is complete.
	\end{proof}

	\begin{remark}
		Since the principal invariants $I_k, k=1,2,3$ are the coefficients of the characteristic polynomial and $I_1(C_e)=I_1(C\,C_p^{-1})$, $I_2(C_e)=I_2(C\,C_p^{-1})$, $I_3(C_e)=I_3(C\,C_p^{-1})$, the eigenvalues of $C_e$ and $C\,C_p^{-1}$ coincide. Clearly,
		$C_e\in {\rm PSym}(3)$, however $C\,C_p^{-1}\not\in{\rm Sym}(3)$ in general, unless $C$ and $C_p^{-1}$ commute.
	\end{remark}

	\begin{lemma}
		\label{lemmasigmatild}
		The introduced stress tensors $\Sigma_e, \widetilde{\Sigma}, \tau_e$ are related as follows
		\begin{align*}
			\Sigma_e=F_p^{-T} \widetilde{\Sigma}\,F_p^T, \qquad\quad
			\widetilde{\Sigma}=F^{T} \tau_e\,F^{-T}.
		\end{align*}
	\end{lemma}
	\begin{proof}
		For arbitrary increment  $H\in \R^{3\times 3}$, we compute
		\begin{align*}
			\langle D_{F}[W({F_e})],H\rangle&=\langle D_F[W(F\,F_p^{-1})],H\rangle=\langle D_{F_e}[W({F_e})],H\,F_p^{-1}\rangle\notag\\&=\langle D_{F_e}[W({F_e})]\,F_p^{-T},H\rangle.
		\end{align*}
		On the other hand, we deduce
		\begin{align*}
			\langle D_F[\widetilde{W}(C\,C_p^{-1})],H\rangle&=\langle D_F[\widetilde{W}(F^TF\,C_p^{-1})],H\rangle
			\\
			&=\langle D[\widetilde{W}(C\,C_p^{-1})],F^THC_p^{-1}+H^TF\,C_p^{-1}\rangle\notag\\&=2\,\langle F\,{\rm sym}[D[\widetilde{W}(C\,C_p^{-1})]C_p^{-1}],H\rangle,\notag
		\end{align*}
		for all $H\in \R^{3\times 3}$.
		In view of Lemma \ref{lemaplasticn} we have $W(F_e)=\widetilde{W}(C\,C_p^{-1})$. Therefore, we obtain
		\begin{align*}
			2\,F\,{\rm sym}[D[\widetilde{W}(C\,C_p^{-1})]C_p^{-1}]=D_{F_e}[W({F_e})]\,F_p^{-T},
		\end{align*}
		and further
		\begin{align*}
			F_e^TD_{F_e}[W({F_e})]\,F_p^{-T}&=2\,F_e^TF\,{\rm sym}[D[\widetilde{W}(C\,C_p^{-1})]C_p^{-1}]=
			2\,F_p^{-T}C\,{\rm sym}[D[\widetilde{W}(C\,C_p^{-1})]C_p^{-1}].
		\end{align*}
		The above relation implies
		\begin{align*}
			\Sigma_e=F_e^TD_{F_e}[W({F_e})]&=
			2\,F_p^{-T}C\,D_C[\widetilde{W}(C\,C_p^{-1})]\,F_p^{T}=\,F_p^{-T}
			\widetilde{\Sigma}\,F_p^{T}.\notag
		\end{align*}
		Therefore, using Remark \ref{plastrem1} the proof is complete.
	\end{proof}
	Next, we introduce a helpful lemma.
	\begin{lemma}\label{lemmaPSym}
		If $t\mapsto C_p(t)\in \R^{3\times 3}$ is continuous and satisfies:
		\begin{equation*}
			\left.
			\begin{array}{rll}
				\det C_p(t)\!\!&=1 \quad\text{ for all}\quad t>0,\vspace{1mm}\\
				C_p(0)\!\!&\in{\rm PSym}(3),\vspace{1mm}\\
				C_p(t)\!\!&\in{\rm Sym}(3)\quad\text{ for all}\quad t>0
			\end{array}
			\right\}\quad \Rightarrow \quad C_p(t)\in{\rm PSym}(3) \quad\text{ for all}\quad t>0.
		\end{equation*}
	\end{lemma}
	
	\begin{proof} Using  Cardano's formula and due to the symmetry of $C_p$, the continuity of the map $t\mapsto C_p(t)$ implies the continuity of
		mappings  $t\mapsto \lambda_i(t)$, $i=1,2,3$, where $\lambda_i(t)\in \R$ are the eigenvalues of $C_p(t)$. Since $\lambda_i(0)>0$ and
		$\lambda_1(t)\lambda_2(t)\lambda_3(t)=1$ for all $t>0$, it follows that $\lambda_i(t)>0$ for all $t>0$ and the proof is complete.
	\end{proof}

	We can slightly weaken the assumption in the previous lemma: $\det C_p(t)>0$ for all $t>0$ is sufficient.
	
	\section{The Simo-Miehe 1992 spatial model}\label{simosection}
	\setcounter{equation}{0}
	
	In the remainder of this paper we  discuss different proposal from the literature for plasticity models in $C_p$. Simo \cite{simo1993recent} (see also Reese and Wriggers \cite{Reese97a} and Miehe \cite[page 72, Prop. 5.25]{Miehe92}) considered the spatial flow rule  in the form
	\begin{align}\label{simoflow1}
		-\frac{1}{2}\,\mathcal{L}_v(B_e)=\lambda^+_{\rm p}\, \Partial _{\tau_e} \Phi(\tau_e)\cdot B_e,
	\end{align}
	where  the Lie-derivative $\mathcal{L}_v(B_e)$ is given by
	$
	\mathcal{L}_v(B_e):=F\, \frac{\rm d}{\rm dt}[C_p^{-1}]\, F^T\in {\rm Sym}(3)
	$,  the tensor $
	\tau_e=2\, \Partial _{B_e} W(B_e)\cdot B_e
	$  is the  symmetric Kirchhoff stress tensor, the yield function
	$
	\Phi(\tau_e)=\|\dev_3\tau_e\|-\sqrt{\frac{2}{3}}\sigma_{\textbf{y}}
	$
	and the plastic multiplier $\lambda^+_{\rm p}$ satisfies the  Karush-Kuhn-Tucker (KKT)-optimality constraints
	\begin{align}\label{subgama}
		\lambda^+_{\rm p}\geq 0, \qquad \Phi(\tau_e)\leq0, \qquad \lambda^+_{\rm p}\, \Phi(\tau_e)=0.
	\end{align}
	The flow rule \eqref{simoflow1} is equivalent with
	\begin{align}\label{frsm2}
		\frac{\rm d}{\rm dt}[C_p^{-1}]&=-2\,\lambda^+_{\rm p} \, F^{-1}[\Partial _{\tau_e} \Phi(\tau_e)\cdot B_e]\, F^{-T}\notag\\&=-2\,\lambda^+_{\rm p} \, F^{-1}\left[\frac{\dev_3\tau_e}{\|\dev_3\tau_e\|}\cdot B_e\right]\, F^{-T},
	\end{align}
	which, in view of the properties \eqref{subgama} of $\lambda^+_{\rm p}$, can be written with a subdifferential
	\begin{align}\label{frsm3}
		\frac{\rm d}{\rm dt}[C_p^{-1}]\in -2 \, F^{-1}\left[\Partial_{\tau_e} \Chi(\dev_3\tau_e)\cdot B_e\right]\, F^{-T},
	\end{align}
	where $\Chi$ is the indicator function of the elastic domain
	\begin{align*}\mathcal{E}_{\rm e}(\tau_{_{\rm e}},\frac{2}{3}\, {\boldsymbol{\sigma}}_{\!\mathbf{y}}^2)=\left\{ \tau_e\in {\rm Sym}(3) \big|\,\ \|\dev_3
		\tau_e\|^2\leq\frac{2}{3}\, {\boldsymbol{\sigma}}_{\!\mathbf{y}}^2\right\}=\{\tau_e\in {\rm Sym}(3)\, |\, \Phi(\tau_e)\leq0\}.
	\end{align*}
	The subdifferential $\Partial \Chi (\dev_3\tau_{e})$ of the indicator function $\Chi$ is the normal cone
	\begin{align*}
		\mathcal{N}(\mathcal{E}_{\rm e}({\tau_{e}},\frac{2}{3}\,{\boldsymbol{\sigma}}_{\!\mathbf{y}}^2);\dev_3 {\tau}_e)=\left\{\begin{array}{ll}
			0, & {\tau}_e\in {\rm int}(\mathcal{E}_{\rm e}({\tau_{e}},\frac{2}{3}{\boldsymbol{\sigma}}_{\!\mathbf{y}}^2))\vspace{2mm}\\
			\{\lambda^+_{\rm p}\, \frac{\dev_3\tau_{e}}{\|\dev_3\tau_{e}\|}\,|\, \lambda^+_{\rm p}\in \R_+\},& {\tau}_e\not\in {\rm int}(\mathcal{E}_{\rm e}({\tau_{e}},\frac{2}{3}{\boldsymbol{\sigma}}_{\!\mathbf{y}}^2)).
		\end{array}\right.
	\end{align*}
	
	We deduce (see the model Eq. (5.25) from \cite{Miehe92}) an equivalent definition for $\mathcal{L}_v(B_e)$ given by
	\begin{align*}
		-\frac{1}{2}\,\mathcal{L}_v(B_e)=\lambda^+_{\rm p}\, \frac{\dev_3\tau_e}{\|\dev_3\tau_e\|}\cdot B_e.
	\end{align*}
	Since $C_p=F^T\, B_e^{-1}\, F$
	we have
	$
	\mathcal{L}_v(B_e)=F\, \frac{\rm d}{\rm dt}[C_p^{-1}]\, F^T=F\left(\, \frac{\rm d}{\rm dt}[F^{-1}B_eF^{-T}]\right)\, F^T.
	$ On the other hand, from \eqref{frsm2} it follows
	that
	\begin{align}\label{frsm4}
		\frac{\rm d}{\rm dt}[C_p^{-1}]\, C_p&=-2\,\lambda^+_{\rm p} \, F^{-1}\left[\frac{\dev_3\tau_e}{\|\dev_3\tau_e\|}\cdot B_e\right]\, F^{-T} \,F^T\, B_e^{-1}\, F\notag\\&\in-2\, F^{-1}\,\Partial_{\tau_e}\Chi(\dev_3\tau_e) \, F.
	\end{align}
	Since
	\begin{align}\label{derivaredet}
		\frac{\rm d}{{\rm d t}}[\det {C}_p^{\,-1}]&=\langle {\rm Cof}\,{C}_p^{\,-1}, \frac{\rm d}{{\rm d t}}[ {C}_p^{\,-1}]\rangle=\det {C}_p^{\,-1}\langle {C}_p,
		\frac{\rm d}{{\rm d t}}[ {C}_p^{\,-1}]\rangle\notag\\&
		=\det {C}_P^{\,-1}\langle \id, \frac{\rm d}{{\rm d t}}[ {C}_p^{\,-1}] {C}_p\rangle,
	\end{align}
	from the flow rule \eqref{frsm2} together with $\det C_p(0)=1$ and $\tr(F^{-1}\dev_3\tau_e\, F)=0$ it follows at once that
	$
	\det C_p(t)=1  \text{ for all }  t\geq 0.
	$
	
	\medskip
	
	The next step is to prove that the flow rule   \eqref{simoflow1} implies $\frac{\rm d}{{\rm d t}}[W(F_e)]\leq 0$ at fixed $F$, i.e. the reduced
	dissipation inequality is satisfied. We compute for fixed in time $F$
	\begin{align}\label{derWp}
		\frac{\rm d}{{\rm d t}}[W(F F^{-1}_p)]&=\langle D_{F_e} W(F_e),F\frac{\rm d}{{\rm d t}}[F^{-1}_p]\rangle=\langle D_{F_e} W(F_e),FF^{-1}_pF_p\frac{\rm
			d}{{\rm d t}}[F^{-1}_p]\rangle
		\\
		&=\langle F_e^T D_{F_e} W(F_e),F_p\frac{\rm d}{{\rm d t}}[F^{-1}_p]\rangle=\langle \Sigma_{e},F_p\frac{\rm d}{{\rm d t}}[F^{-1}_p]\rangle\notag\\&=-\langle \Sigma_{e},\underbrace{{\rm sym}(\frac{\rm d}{{\rm d
					t}}[F_p]F_p^{-1})}_{D_p}\rangle,\notag
	\end{align}
	since $\Sigma_{e}\in{\rm Sym}(3)$. We also have
	\begin{align*}
		\frac{\rm d}{{\rm d t}}[C_p]&=\frac{\rm d}{{\rm d t}}[F_p^TF_p]=\frac{\rm d}{{\rm d t}}[F_p^T]F_p+F_p^T\frac{\rm d}{{\rm d
				t}}[F_p]\notag\\&=F_p^T\left(F_p^{-T}\frac{\rm d}{{\rm d t}}[F_p^T]\right)F_p+
		F_p^T\left(\frac{\rm d}{{\rm d t}}[F_p]F_p^{-1}\right)F_p=2\,F_p^TD_pF_p,\notag
	\end{align*}
	where
	$
	D_p:={\rm sym} \left(\frac{\rm d}{{\rm d t}}[F_p]F_p^{-1}\right).
	$
	Hence, we easily deduce the representation
	$
	D_p=\frac{1}{2}\,F_p^{-T}\frac{\rm d}{{\rm d t}}[C_p]F_p^{-1}.
	$
	Therefore, with \eqref{derWp} we obtain
	\begin{align}\label{derivarethermo0}
		\frac{\rm d}{{\rm d t}}[W(F F^{-1}_p)]&=-\langle \Sigma_{e},\frac{1}{2}\,F_p^{-T}\frac{\rm d}{{\rm d t}}[C_p]F_p^{-1}\rangle.
	\end{align}
	Moreover, since $\Sigma_{e}=F_e^T\tau_e \, F_e^{-T}$, we deduce
	\begin{align}\label{derivarethermo01}
		\frac{\rm d}{{\rm d t}}[W(F F^{-1}_p)]&=-\frac{1}{2}\langle F_e^T\tau_e \, F_e^{-T},\,F_p^{-T}\frac{\rm d}{{\rm d t}}[C_p]F_p^{-1}\rangle=\frac{1}{2}\langle F_e^T\tau_e \, F_e^{-T},\,F_p\frac{\rm d}{{\rm d t}}[C_p^{-1}]F_p^{T}\rangle\notag\\&=\frac{1}{2}\langle \,F_p^T\,F_e^T\tau_e \, F_e^{-T}F_p,\frac{\rm d}{{\rm d t}}[C_p^{-1}]\rangle=\frac{1}{2}\langle \,F^T\tau_e \, B_e^{-1}\,F,\frac{\rm d}{{\rm d t}}[C_p^{-1}]\rangle\\&=\frac{1}{2}\langle \tau_e,F\,\frac{\rm d}{{\rm d t}}[C_p^{-1}]\, F^T\, B_e^{-1}\rangle.\notag
	\end{align}
	The flow rule \eqref{frsm2} implies
	\begin{align}\label{derivarethermo02}
		\frac{\rm d}{{\rm d t}}[W(F F^{-1}_p)]&=-\lambda^+_{\rm p} \,\langle \tau_e,\frac{\dev_3\tau_e}{\|\dev_3\tau_e\|}\rangle=-\lambda^+_{\rm p} \, \|\dev_3\tau_e\|\leq 0.
	\end{align}

	In view of the definition of $\Sigma_{e}=F_e^T\, \tau_e\, F_e^{-T}$ we have $F^{-1}[\,{\tau_e} \, B_e ] F^{-T}=F_p^{-1}[{\Sigma_{e} } ] F_p^{-T}$. For the isotropic case it holds $\tau_e\, B_e=B_e \, \tau_e$. Hence,
	\begin{align*}F^{-1}[{\tau_e} \, B_e ] F^{-T}&=F^{-1}[ \, B_e {\tau_e}] F^{-T}=F_p^{-1}F_e^{-1}[ \, F_e\, F_e^T \tau_e] F_e^{-T}F_p^{-T}\\
		&=F_p^{-1}[ \, F_e^T \tau_e F_e^{-T}]F_p^{-T}=F_p^{-1}[ \Sigma_{e}]F_p^{-T}. \notag
	\end{align*}
	We also observe   $F^{-1}[\,{\rm tr}({\tau_e}) \, B_e ] F^{-T}=F_p^{-1}[{\rm tr}({\Sigma_{e} }) ] F_p^{-T}$.
	Thus, we obtain
	$$
	F^{-1}[{\dev_3 \tau_e} \, B_e ] F^{-T}=F_p^{-1}[{\dev_3 \Sigma_{e} } ] F_p^{-T}.
	$$
	Together with  Remark \ref{remarktauesigmae} this implies that \begin{align}\label{difsimomiehe}
		F^{-1}\left[\frac{\dev_3 \tau_e}{\|\dev_3 \tau_e\|}  \, B_e \right] F^{-T}=F_p^{-1}\left[\frac{\dev_3 \Sigma_{e} }{\|\dev_3 \Sigma_{e} \|} \right ] F_p^{-T}.
	\end{align}
	Therefore, in the isotropic case, the flow rule \eqref{frsm3}  has a subdifferential structure: \begin{align}\label{frsm300}
		\frac{\rm d}{\rm dt}[C_p^{-1}]\in -2 \, F_p^{-1}\,[\Partial_{\Sigma_e} \Chi(\dev_3\Sigma_e)]\, F_p^{-T},
	\end{align}
	where $\Chi$ is the indicator function of the elastic domain
	$$\mathcal{E}_{\rm e}(\Sigma_{_{\rm e}},\frac{2}{3}\, {\boldsymbol{\sigma}}_{\!\mathbf{y}}^2)=\left\{ \Sigma_e\in {\rm Sym}(3) \big|\,\ \|\dev_3
	\Sigma_e\|^2\leq\frac{2}{3}\, {\boldsymbol{\sigma}}_{\!\mathbf{y}}^2\right\}.
	$$

	In view of the above equivalent representations of the flow rule, we may summarize the properties  of the Simo-Miehe 1992 model:
	\begin{itemize}
		\item[i)] from \eqref{frsm2} it follows, in the isotropic case (in which $\tau_e$ and $B_e$ commute), that $C_p(t)\in{\rm Sym}(3)$;
		\item[ii)] plastic incompressibility: from  \eqref{frsm4} and \eqref{derivaredet} it follows that
		$\det C_p(t)=1$, since the right hand side is trace-free;
		\item[iii)] for the isotropic case, the right hand-side of \eqref{frsm2} is a function of $C_p^{-1}$ and $C$ alone, since $  B_e=F\,C_p^{-1}\, F^{T}$  and $F^{-1} B_e F=F^{-1}F\,C_p^{-1}\, F^{T} F=C_p^{-1}\, C$;
		\item[iv)] from i) and ii) together and using Lemma \ref{lemmaPSym} it follows that $C_p(t)\in{\rm PSym}(3)$;
		\item[v)] it is thermodynamically correct;
		\item[vi)] the right hand side in the representation  \eqref{frsm2} is not the subdifferential of the indicator function of some convex domain in some stress space. However, this model is an associated plasticity model in the isotropic case, see Proposition  \ref{eqlionsimo} and Proposition \ref{propositionegLion}.
	\end{itemize}

	\section{The Miehe 1995 referential model}\label{Miehe1995}\setcounter{equation}{0}

	Shutov \cite{Shutovpers} interpreted that Miehe in \cite{Miehe95} considered the flow rule\footnote{Miehe \cite{Miehe95} only defines the elastic domain $
		\mathcal{E}_{\rm e}(\widetilde{\Sigma},{\frac{2}{3}}\, \sigma_{\textbf{y}}^2):=\left\{\widetilde{\Sigma}\in \R^{3\times 3}\, \Big| \, \tr( (\dev_3 \widetilde{\Sigma})^2)\leq {\frac{2}{3}}\, \sigma_{\textbf{y}}^2\right\}
		$ in terms of $\tau_e$, i.e. $\mathcal{E}_{\rm e}(\tau_{_{\rm e}},\frac{2}{3}\, {\boldsymbol{\sigma}}_{\!\mathbf{y}}^2)=\left\{ \tau\in {\rm Sym}(3) \big|\,\ \|\dev_3
		\tau\|^2\leq\frac{2}{3}\, {\boldsymbol{\sigma}}_{\!\mathbf{y}}^2\right\}
		$. He uses the same notation for the referential quantities. Therefore, we have two interpretations at hand $
		\Phi(\widetilde{\Sigma})=\|\dev_3\tau_e\|-{\frac{2}{3}}\, \sigma_{\textbf{y}}^2=\sqrt{\tr( (\dev_3 \widetilde{\Sigma})^2)}-{\frac{2}{3}}\, \sigma_{\textbf{y}}^2$. On the other hand, in the isotropic case, we have also $\Phi(\widetilde{\Sigma})=\|\dev_3\Sigma_e\|-{\frac{2}{3}}\, \sigma_{\textbf{y}}^2.
		$
	}
	\begin{align}\label{frSMi}
		\frac{\rm d}{\rm dt} [C_p^{-1}]\, C_p=-\dd\lambda_{\rm p}^+D_ {\widetilde{\Sigma}} \Phi(\widetilde{\Sigma}),
	\end{align}
	where $\widetilde{\Sigma}=2\,C\, D_C[\widetilde{W}(C\, C_p^{-1})]$  and
	$$
	\Phi(\widetilde{\Sigma})=\|\dev_3\tau_e\|-\sqrt{\frac{2}{3}}\, \sigma_{\textbf{y}}=\sqrt{\tr( (\dev_3 \widetilde{\Sigma})^2)}-\sqrt{\frac{2}{3}}\, \sigma_{\textbf{y}}.
	$$
	
	In this model, it is important tp note that it is not the Frobenius norm of $\dev_3 \widetilde{\Sigma}$ which is used in the yield function $\Phi$.  Instead, in the denominator $\mathcal{F}:=\sqrt{\tr( (\dev_3 \widetilde{\Sigma})^2)}$  is  considered, see Eq. (52) from \cite{shutov2008finite}. Since $\dev_3 \widetilde{\Sigma}\not\in {\rm Sym}(3)$, it follows that $\mathcal{F}:=\sqrt{\tr( (\dev_3 \widetilde{\Sigma})^2)}\neq \|\dev_3 \widetilde{\Sigma}\|$\,.  Indeed, we have
	\begin{align*}
		\sqrt{\tr[(\dev_3 \widetilde{\Sigma})^2]}&=\| \dev_3( \widetilde{\Sigma})\|\quad \Leftrightarrow\quad \langle \dev_3\widetilde{\Sigma},(\dev_3\widetilde{\Sigma})^T\rangle =\langle \dev_3\widetilde{\Sigma},\dev_3\widetilde{\Sigma}\rangle\notag
		\\
		& \Leftrightarrow\quad
		\langle \dev_3\widetilde{\Sigma},{\rm skew}(\dev_3 \widetilde{\Sigma})\rangle=0\quad \Leftrightarrow\quad \dev_3\widetilde{\Sigma}\in {\rm Sym}(3)\notag\\& \Leftrightarrow\quad\widetilde{\Sigma}\in {\rm Sym}(3).
	\end{align*}
	For  the simplest Neo-Hooke elastic energy considered in Appendix A.2, $W(F_e)=\tr(C_e)=\widetilde{W}(C\, C_p^{-1})=\frac{1}{2}\,\tr(C\, C_p^{-1})$, we have $\widetilde{\Sigma}=C\, C_p^{-1},
	$ which is not symmetric.
	Hence $\sqrt{\tr[(\dev_3 \widetilde{\Sigma})^2]}\neq \| \dev_3 \widetilde{\Sigma}\|$.
	Let us again remark that $\widetilde{\Sigma}$ is not necessarily symmetric for general $C_p$. However, using  Lemma \ref{lemmasigmatild}, we deduce
	\begin{align}\label{simetrictild0}
		\widetilde{\Sigma}\,C_p&=F_p^{T} \Sigma_e\,F_p^{-T}\,C_p= F_p^{T} \Sigma_e\,F_p\in {\rm Sym}(3)\quad \Rightarrow \quad \dev_3 \widetilde{\Sigma}\cdot C_p\in {\rm Sym}(3),\\
		C_p^{-1}\, \widetilde{\Sigma}&=C_p^{-1}\,F_p^{T} \Sigma_e\,F_p^{-T}= F_p^{-1} \Sigma_e\,F_p^{-T}\in {\rm Sym}(3)\quad \Rightarrow \quad C_p^{-1}\, \dev_3 \widetilde{\Sigma}\in {\rm Sym}(3).\notag
	\end{align}
	In the following, we discuss first the sign  of the quantity\footnote{If we are not looking for the sign of $\tr( (\dev_3 \widetilde{\Sigma})^2)$ for all $\dev_3 \widetilde{\Sigma}\in \R^{3\times 3}$, then considering two particular values of $\dev_3 \widetilde{\Sigma}$, e.g.
		\begin{align*}
			\dev_3 \widetilde{\Sigma}=\left(
			\begin{array}{ccc}
				-\frac{1}{2} & 1 & 2 \\
				-2 & -\frac{1}{2} & 3 \\
				-1 & -3 & -\frac{1}{2} \\
			\end{array}
			\right)\qquad \text{and} \qquad \dev_3 \widetilde{\Sigma}=\left(
			\begin{array}{ccc}
				-\frac{1}{3} & 0 & 0 \\
				0 & \frac{2}{3} & 0 \\
				0 & 0 & -\frac{1}{3} \\
			\end{array}
			\right),
		\end{align*}
		we obtain  $\tr[(\dev_3 \widetilde{\Sigma})^2]=-2$ and $\tr[(\dev_3 \widetilde{\Sigma})^2]=\frac{2}{3}$, respectively.  Hence, $\tr[(\dev_3 \widetilde{\Sigma})^2]$ is not positive for all $\widetilde{\Sigma}\in \R^{3\times 3}$. }
	$\mathcal{F}^2:={\tr( (\dev_3 \widetilde{\Sigma})^2)}$. First, we deduce
	\begin{align}\label{ShuIhl00}
		\tr[(\dev_3 \widetilde{\Sigma})^2]&=\langle(\dev_3 \widetilde{\Sigma})\,  (\dev_3 \widetilde{\Sigma}),\id\rangle=\langle \widetilde{\Sigma}\,  (\dev_3 \widetilde{\Sigma}),\id\rangle=\langle C_p^{-1}\,\widetilde{\Sigma}\,  (\dev_3 \widetilde{\Sigma})\,C_p,\id\rangle\notag\\&=\langle C_p^{-1}\,\widetilde{\Sigma}\, ( \dev_3 \widetilde{\Sigma}\cdot C_p)^T,\id\rangle=\langle C_p^{-1}\,\widetilde{\Sigma}\, C_p\, (\dev_3 \widetilde{\Sigma})^T,\id\rangle\\&=\langle C_p^{-1}\,\widetilde{\Sigma}\, C_p, \dev_3 \widetilde{\Sigma}\rangle.\notag
	\end{align}
	We further see that
	\begin{align}\label{ShuIhl220}
		\langle C_p^{-1}\,\widetilde{\Sigma}\,& C_p, \dev_3 \widetilde{\Sigma}\rangle=
		\langle U_p^{-1}U_p^{-1}\,\widetilde{\Sigma}\, U_p\,U_p, \dev_3 \widetilde{\Sigma}\rangle
		=
		\langle U_p^{-1}\,\widetilde{\Sigma}\, U_p, U_p^{-1}\dev_3 \widetilde{\Sigma}\,U_p\rangle\notag\\&
		=
		\langle U_p^{-1}\,\widetilde{\Sigma}\, U_p, U_p^{-1} \widetilde{\Sigma}\,U_p-\frac{1}{3}\tr(\widetilde{\Sigma})\cdot \id\rangle\\&
		=
		\langle U_p^{-1}\,\widetilde{\Sigma}\, U_p, U_p^{-1} \widetilde{\Sigma}\,U_p-\frac{1}{3}\tr(U_p^{-1}\,\widetilde{\Sigma}\, U_p)\cdot \id\rangle\notag\\&
		=
		\langle U_p^{-1}\,\widetilde{\Sigma}\, U_p, \dev_3(U_p^{-1} \widetilde{\Sigma}\,U_p)\rangle=
		\langle \dev_3( U_p^{-1}\,\widetilde{\Sigma}\, U_p), \dev_3(U_p^{-1} \widetilde{\Sigma}\,U_p)\rangle\notag\\&=
		\| \dev_3( U_p^{-1}\,\widetilde{\Sigma}\, U_p)\|^2\geq 0,\notag
	\end{align}
	where $U_p^2=C_p$. Thus $\mathcal{F}^2$ is positive and $\mathcal{F}$ is well defined.
	
	Since
	$
	\dd D_ {\widetilde{\Sigma}} \Phi(\widetilde{\Sigma})=\frac{1}{\sqrt{\tr[(\dev_3 \widetilde{\Sigma})^2]}}\, (\dev_3 \widetilde{\Sigma})^T
	$
	the flow rule \eqref{frSMi} becomes
	\begin{align}\label{shutovremark}
		\frac{\rm d}{\rm dt} [C_p^{-1}]\, C_p&=- \frac{\lambda_{\rm p}^+}{\sqrt{\tr[(\dev_3 \widetilde{\Sigma})^2]}}\, (\dev_3 \widetilde{\Sigma})^T\notag\\& \Leftrightarrow\qquad
		C_p\,\frac{\rm d}{\rm dt} [C_p^{-1}]=-\frac{\lambda_{\rm p}^+}{\sqrt{\tr[(\dev_3 \widetilde{\Sigma})^2]}}\, \dev_3 \widetilde{\Sigma}\,,
	\end{align}
	Further, in view of \eqref{simetrictild0}, we obtain
	\begin{align}\label{echimiehehelm}
		\frac{\rm d}{\rm dt} [C_p^{-1}]&=-\frac{\lambda_{\rm p}^+}{\sqrt{\tr[(\dev_3 \widetilde{\Sigma})^2]}}\,C_p^{-1}\, \dev_3 \widetilde{\Sigma}
		\notag\\& \Leftrightarrow\qquad
		\frac{\rm d}{\rm dt} [C_p]=\frac{\lambda_{\rm p}^+}{\sqrt{\tr[(\dev_3 \widetilde{\Sigma})^2]}}( \dev_3 \widetilde{\Sigma})\, C_p\in {\rm Sym}(3).
	\end{align}
	Using  Lemma \ref{lemmaPSym} we obtain  that $C_p\in{\rm PSym}(3)$.
	
	We remark that the flow rule considered by Miehe \cite{Miehe95} (in this interpretation)  coincides with the flow rule \eqref{SHflow} considered by Helm \cite{helm2001formgedachtnislegierungen}, see Proposition  \ref{prophelmmiehe}.

	\begin{remark}\label{remarkconvexityMi}
		Although  the flow rule considered in this  interpretation of the  Miehe 1995 model \cite{Miehe95} has a subdifferential structure, the yield-function $\Phi$ is not convex. Hence, the flow rule is not a convex flow rule. In order to see the non-convexity of $\Phi(\widetilde{\Sigma})$ we observe first by looking at sublevel-sets that
		$$
		\Phi(\widetilde{\Sigma})=\sqrt{{\rm tr}( (\dev_3 \widetilde{\Sigma})^2)}-{\frac{2}{3}}\, \sigma_{\textbf{y}}^2 \  \text{is convex} \quad \Leftrightarrow\quad \widetilde{\Phi}(\widetilde{\Sigma})={\rm tr}[(\dev_3 \widetilde{\Sigma})^2]\  \text{is convex}.
		$$
		The second derivative for the simpler function $\widetilde{\Phi}(\widetilde{\Sigma})$ is
		$$
		\dd D^2_ {\widetilde{\Sigma}} \widetilde{\Phi}(\widetilde{\Sigma}).(H,H) = \langle (\dev_3 H)^T, \dev_3 H\rangle={\rm tr} [(\dev_3 H)^2],\qquad \forall\ \, \widetilde{\Sigma}, \, H\in \R^{3\times 3}.
		$$
		We know that ${\rm tr} [(\dev_3 H)^2]$ is not positive for all $ H\in \R^{3\times 3}$, since for the previous considered matrix $H$, such that
		$$
		\dev_3 H=\left(
		\begin{array}{ccc}
		-\frac{1}{2} & 1 & 2 \\
		-2 & -\frac{1}{2} & 3 \\
		-1 & -3 & -\frac{1}{2} \\
		\end{array}
		\right),
		$$
		we obtain  ${\rm tr}[(\dev_3 H)^2]=-2$. Therefore $\widetilde{\Phi}(\widetilde{\Sigma})$ is not convex, and thus $\Phi(\widetilde{\Sigma})$ cannot be convex.
	\end{remark}

	\section{The Lion 1997 multiplicative elasto-plasticity formulation in terms of the plastic metric $C_p=F_p^TF_p$}\setcounter{equation}{0}
	\label{appendixplasticity}

	This derivation was given by Lion \cite[Eq. (47.2)]{lion1997physically} in the general form (see also \cite[Eq. (6.33)]{helm2001formgedachtnislegierungen}) and by Dettmer-Reese \cite{dettmer2004theoretical} in the isotropic case. Following   \cite{dettmer2004theoretical} we consider   a perfect plasticity model for the plastic metric $C_p$ based on the flow rule
	\begin{align}\label{choicchi00}
		\frac{\rm d}{{\rm d t}}[C_p^{-1}]\in-F_p^{-1}\,\Partial \Chi (\dev_3\Sigma_{e})\,F_p^{-T}\in {\rm Sym}(3) \qquad \text{for} \quad \Sigma_{e}\in {\rm Sym}(3).
	\end{align}
	
	Again, in this model it is not clear from the outset, that it is a formulation in $C_p$ alone.  The goal of such a 6-dimensional formulation is to avoid any explicit  computation of the plastic distortion $F_p$. However, the right hand side of the above proposed flow rule is, in fact, a multivalued function in $C$ and $C_p^{-1}$ alone. Hence, we can express  the flow rule \eqref{choicchi00} entirely
	in the form\footnote{Note carefully, that $ f(C,C_p^{-1})$ is not necessarily symmetric. Moreover $C_p^{-1}\widehat{f}(C,C_p^{-1})\not\in {\rm Sym}(3)$ in general.}
	\begin{align}\label{choicchi0f}
		\frac{\rm d}{{\rm d t}}[C_p^{-1}]\,C_p\in f(C,C_p^{-1}).
	\end{align}
	In order to show this remarkable property (satisfied only for isotropic response), and to determine the explicit form of the function $f(C,C_p^{-1})$, in view of \eqref{isoplastalpha} we remark that \
	\begin{align*}
		F_p^{-1}\frac{\dev_3\Sigma_{e}}{\|\dev_3\Sigma_{e}\|}F_p^{-T}&=\frac{1}{\|\dev_3\Sigma_{e}\|}\left[F_p^{-1}\Sigma_{e}F_p^{-T}-\frac{1}{3}\,\tr(\Sigma_{e})\,
		C_p^{-1}\right],\notag\\ \tr(\Sigma_{e})&=2\,\langle \id, C_e\, (\alpha_1\, \id +\alpha_2\, C_e+\alpha_3 \, C_e^2),\\
		\|\Sigma_{e}\|^2&=4\langle C_e\, (\alpha_1\, \id +\alpha_2\, C_e+\alpha_3 \, C_e^2), C_e\, (\alpha_1\, \id +\alpha_2\, C_e+\alpha_3 \, C_e^2)\rangle,\notag\\
		\|\dev_3\Sigma_{e}\|&=\sqrt{\|\Sigma_{e}\|^2-\frac{1}{3}[\tr(\Sigma_{e})]^2}.\notag
	\end{align*}
	It is clear that
	$
	F_p^{-1}{\Sigma_{e}}F_p^{-T}=2\,F_p^{-1}\, C_e\, (\alpha_1\, \id +\alpha_2\, C_e+\alpha_3 \, C_e^2)\,F_p^{-T}\in{\rm Sym}(3),
	$
	and
	\begin{align}
		C_e&=F_e^TF_e=F_p^{-T}F^T F\, F_p^{-1}=F_p^{-T}C\, F_p^{-1},\notag\\
		\Sigma_{e}&= 2\,F_p^{-T}(\alpha_1\,C +\alpha_2\,C\,C_p^{-1}C+\alpha_3 \,C\, C_p^{-1}C\, C_p^{-1}C)\,F_p^{-1},\notag\\
		\tr(\Sigma_{e})&=2\,\tr(\alpha_1\,C\,C_p^{-1} +\alpha_2\,C\,C_p^{-1}C\,C_p^{-1}+\alpha_3 \,C\, C_p^{-1}C\, C_p^{-1}C\,C_p^{-1}),\notag\\
		\|\Sigma_{e}\|^2&=4\langle C_p^{-1}\widehat{f},\widehat{f}C_p^{-1}\rangle.\notag
	\end{align}
	where
	\begin{align*}
		\widehat{f}:=\alpha_1\,C +\alpha_2\,C\,C_p^{-1}C+\alpha_3 \,C\, C_p^{-1}C\, C_p^{-1}C.
	\end{align*}
	Hence, we deduce
	\begin{align}
		F_p^{-1}{\Sigma_{e}}F_p^{-T}&=2\,C_p^{-1}\,\widehat{f}(C,C_p^{-1})\,C_p^{-1}\in{\rm Sym}(3),\\
		\tr(\Sigma_{e})&=2\,\tr(\widehat{f}(C,C_p^{-1})\,C_p^{-1}),\qquad
		\|\Sigma_{e}\|^2=4\langle C_p^{-1}\,\widehat{f}(C,C_p^{-1}),\widehat{f}(C,C_p^{-1})\,C_p^{-1}\rangle,\notag\\
		\|\dev_3\Sigma_{e}\|&=2\,\sqrt{\tr[(\widehat{f}(C,C_p^{-1})\,C_p^{-1})^2]-\frac{1}{3}[\tr(\widehat{f}(C,C_p^{-1})C_p^{-1})]^2},\notag
	\end{align}
	where
	\begin{align}
		\widehat{f}(C,C_p^{-1}):=\alpha_1\,C +\alpha_2\,C\,C_p^{-1}C\,+\alpha_3 \,C\, C_p^{-1}C\, C_p^{-1}C\in{\rm Sym}(3),
	\end{align}
	and $\alpha_i=\alpha_i(I_1(C_e),I_2(C_e),I_3(C_e))$, according to \eqref{isoplastalpha}. Therefore,  the multivalued function $f(C,C_p^{-1})$  is given by
	\begin{align}\label{explf}
		{f}(C,C_p^{-1})&=\Bigg\{\frac{-\lambda^+_{\rm p}\,\dev_3(\,C_p^{-1}\widehat{f}(C,C_p^{-1}))}
		{\sqrt{\tr[(\widehat{f}(C,C_p^{-1})C_p^{-1})^2]-\frac{1}{3}[\tr(\widehat{f}(C,C_p^{-1})C_p^{-1})]^2}}
		\ \big| \ \lambda^+_{\rm p}\in \R_+\Bigg\}.
	\end{align}
	In Appendix A.1 we  give the specific expression for the functions $f(C,C_p^{-1})$ and $\widehat{f}(C,C_p^{-1})$ in case of  the Neo-Hooke energy.
	
	On the other hand, in view of equation \eqref{choicchi00} we also have
	\begin{align*}
		\frac{\rm d}{{\rm d t}}[C_p]\,C_p^{-1}=-C_p\frac{\rm d}{{\rm d t}}[C_p^{-1}]\in C_p\,F_p^{-1}\,\Partial \Chi (\dev_3\Sigma_{e})\,F_p^{-T}= F_p^{T}\Partial \Chi
		(\dev_3\Sigma_{e})F_p^{-T}.
	\end{align*}
	
	\noindent
	Hence, it follows that
	\begin{align}\label{f425}
		\frac{\rm d}{{\rm d t}}[C_p]\in F_p^{T}\,\Partial \Chi (\dev_3\Sigma_{e})\,F_p^{-T}C_p=F_p^{T}\Partial \Chi (\dev_3\Sigma_{e})F_p\,\in\, {\rm Sym}(3),
	\end{align}
	which establishes symmetry of $C_p$ whenever $C_p(0)\in{\rm Sym}(3)$.
	
	Another important question is whether the solution $C_p$ of the flow rule \eqref{choicchi00} is such that $\det C_p(t)=1$, for all $t\geq 0$. Let $C_p$ be
	the solution of the flow rule \eqref{choicchi00}. Then, we have
	\begin{align}\label{isused}
		C_p\frac{\rm d}{{\rm d t}}[C_p^{-1}]&=-\lambda^+_{\rm p}\, C_p\,F_p^{-1}\frac{\dev_3\Sigma_{e}}{\|\dev_3\Sigma_{e}\|}\,  F_p^{-T}=
		-\lambda^+_{\rm p}\, F_p^{T}F_pF_p^{-1}\,\frac{\dev_3\Sigma_{e}}{\|\dev_3\Sigma_{e}\|}\,  F_p^{-T}\notag\\&=-\frac{\lambda^+_{\rm p}}{2}\,
		\frac{\dev_3(F_p^{T}\Sigma_{e}\,F_p^{-T})}{\|\dev_3\Sigma_{e}\|},
	\end{align}
	which implies on the one hand
	\begin{equation*}
		\langle \frac{\rm d}{{\rm d t}}[C_p^{-1}]\, C_p,\id\rangle=\langle \frac{\rm d}{{\rm d t}}[C_p^{-1}],C_p\rangle=0.
	\end{equation*}
	On the other hand, the flow rule \eqref{choicchi00} together with $\det\, C_p(0)=1$ leads to
	$
	\det\, C_p(t)=1$, for all $\,t\geq 0.
	$
	
	Let us remark that, in view of \eqref{choicchi0f} and \eqref{explf} we have for the flow rule \eqref{choicchi00}
	\begin{align*}
		\frac{\rm d}{{\rm d t}}[C_p^{-1}]
		=\frac{-\lambda^+_{\rm p}}{\sqrt{\tr[(\widehat{f}(C,C_p^{-1})C_p^{-1})^2]-
				\frac{1}{3}[\tr[\widehat{f}(C,C_p^{-1})C_p^{-1}]]^2}}\,\,
		\underbrace{{\rm dev}_3 [\underbrace{\,C_p^{-1}\widehat{f}(C,C_p^{-1})}_{\not\in {\rm Sym}(3)}]\cdot\, C_p^{-1}}_{\in {\rm Sym}(3)},
	\end{align*}
	which is in concordance with the requirement $C_p\in {\rm Sym}(3)$, as can be seen from \eqref{f425}. Note that the above formula {cannot be read as $$\frac{\rm d}{{\rm d t}}[C_p^{-1}]
		=-\lambda^+_{\rm p}\frac{{\rm dev}_3 {\Sigma}}{\|{\rm dev}_3 {\Sigma}\|}\cdot\, C_p^{-1},$$ for some $\dd{\Sigma}$}, since
	$$[\tr(\widehat{f}(C,C_p^{-1})C_p^{-1})^2]-\frac{1}{3}[\tr(\widehat{f}(C,C_p^{-1})C_p^{-1})]^2\neq \|\dev_3(\widehat{f}(C,C_p^{-1})C_p^{-1})\|^2.$$  To see this, assume to the contrary that equality holds. Then we deduce
	\begin{align}\label{fpmtrdev0}
		[\tr(\widehat{f}(&C,C_p^{-1})C_p^{-1})^2]
		-\frac{1}{3}[\tr(\widehat{f}(C,C_p^{-1})C_p^{-1})]^2
		=\|\dev_3(\widehat{f}(C,C_p^{-1})C_p^{-1})\|^2\\&
		\Leftrightarrow \quad \langle \widehat{f}(C,C_p^{-1})C_p^{-1},(\widehat{f}(C,C_p^{-1})C_p^{-1})^T\rangle= \pm \langle \widehat{f}(C,C_p^{-1})C_p^{-1},\widehat{f}(C,C_p^{-1})C_p^{-1}\rangle.\notag
	\end{align}
	Since $\widehat{f}(C,C_p^{-1})\in {\rm Sym}(3)$, we obtain
	\begin{align}\label{ShuIhl00f}
		\tr[(\widehat{f}(C,C_p^{-1})C_p^{-1})^2]
		&=\langle C_p^{-1}\,\widehat{f}(C,C_p^{-1})C_p^{-1}\,  (\widehat{f}(C,C_p^{-1}),\id\rangle\notag\\&=\langle C_p^{-1}\,\widehat{f}(C,C_p^{-1}), \widehat{f}(C,C_p^{-1})C_p^{-1}\rangle.
	\end{align}
	Using that $C_p\in{\rm PSym}(3)$, we further deduce that
	\begin{align}\label{ShuIhl220f}
		\langle C_p^{-1}\,\widehat{f}(C,C_p^{-1}), \widehat{f}(C,C_p^{-1})C_p^{-1}\rangle&=
		\langle U_p^{-1}\,\widehat{f}(C,C_p^{-1})U_p^{-1}, U_p^{-1} \widehat{f}(C,C_p^{-1})U_p^{-1}\rangle\notag\\&=\|  U_p^{-1}\,\widehat{f}(C,C_p^{-1})U_p^{-1}\|^2,
	\end{align}
	where $U_p^2=C_p$.
	Therefore, from  \eqref{fpmtrdev0} we deduce
	\begin{align}\label{fpmtrdev1} \langle \widehat{f}(&C,C_p^{-1})C_p^{-1},(\widehat{f}(C,C_p^{-1})C_p^{-1})^T\rangle
		=\langle \widehat{f}(C,C_p^{-1})C_p^{-1},
		\widehat{f}(C,C_p^{-1})C_p^{-1}\rangle \notag\\& \Leftrightarrow \quad
		\widehat{f}(C,C_p^{-1})C_p^{-1}\in{\rm Sym}(3),
	\end{align}
	which is not true, in general. However, it is {an associated plasticity model} in the sense of Definition \ref{definitionpld}, see Proposition \ref{propositionegLion}.  We also remark that
	\begin{align}\label{fhat1}
		C_p\,\frac{\rm d}{{\rm d t}}[C_p^{-1}]
		=\frac{-\lambda^+_{\rm p}\,{\rm dev}_3 [\widehat{f}(C,C_p^{-1})\,C_p^{-1}]}{\sqrt{\tr[(\widehat{f}(C,C_p^{-1})C_p^{-1})^2]-
				\frac{1}{3}[\tr[\widehat{f}(C,C_p^{-1})C_p^{-1}]]^2}}\,
		.
	\end{align}

	In conclusion, using Lemma \ref{lemmaPSym}, we have
	\begin{remark}\label{remarkcpPsym}
		Any continuous  solution $C_p\in {\rm Sym}(3)$ of the flow rule \eqref{choicchi00} belongs in fact
		to ${\rm PSym}(3)$.
	\end{remark}
	As for the thermodynamical consistency, we remark that
	\begin{align}\label{ShuIhlDett}
		\frac{\rm d}{\rm dt} [\widetilde{W}(C\, C_p^{-1})]&=\langle D[\widetilde{W}(C\, C_p^{-1})],C\, \frac{\rm d}{\rm dt}[ C_p^{-1}]\rangle=
		\langle C\, D_C[\widetilde{W}(C\, C_p^{-1})]\, C_p, \frac{\rm d}{\rm dt}[ C_p^{-1}]\rangle\notag\\
		&=\frac{1}{2}
		\langle \widetilde{\Sigma}\, C_p, \frac{\rm d}{\rm dt}[ C_p^{-1}]\rangle=\frac{1}{2}
		\langle C_p^{-1}\,\widetilde{\Sigma}\, C_p, C_p \frac{\rm d}{\rm dt}[ C_p^{-1}]\rangle\\
		&=
		-\frac{1}{2}\langle C_p^{-1}\,\widetilde{\Sigma}\,  C_p, \frac{\rm d}{\rm dt}[C_p]\, C_p^{-1}\rangle\notag=
		-\frac{1}{4}\,\frac{\lambda_{\rm p}^+}{\|\dev_3 \widetilde{\Sigma}\|}\langle C_p^{-1}\,\widetilde{\Sigma}\, C_p, \dev_3 \widetilde{\Sigma}\rangle,
	\end{align}
	which, using the formula $ F_p^{T}\Sigma_e\,F_p^{-T}= \widetilde{\Sigma}$, leads to
	\begin{align}\label{ShuIhlDett1}
		\frac{\rm d}{\rm dt} [&\widetilde{W}(C\, C_p^{-1})]=
		-\frac{1}{4}\frac{\lambda_{\rm p}^+}{\|\dev_3 (F_p^{T}\Sigma_e\,F_p^{-T})\|}\langle C_p^{-1}\,F_p^{T}\Sigma_e\,F_p^{-T}\, C_p, \dev_3 ( F_p^{T}\Sigma_e\,F_p^{-T})\rangle\notag
		\\&=
		-\frac{1}{4}\frac{\lambda_{\rm p}^+}{\|\dev_3 (F_p^{T}\Sigma_e\,F_p^{-T})\|}\langle \Sigma_e, \Sigma_e-\frac{1}{3}\tr( \Sigma_e)\cdot \id \rangle
		\\&=
		-\frac{1}{4}\frac{\lambda_{\rm p}^+}{\|\dev_3 (F_p^{T}\Sigma_e\,F_p^{-T})\|}\|\dev_3 \Sigma_e\|^2\leq 0.\notag
	\end{align}
	Note that this proof of thermodynamical consistency may be criticized because it involves the variable $F_p$, which should not appear at all. However, we may also use \eqref{ShuIhl00} and \eqref{ShuIhl220} to obtain
	\begin{align}\label{ShuIhlDett}
		\frac{\rm d}{\rm dt} [\widetilde{W}(C\, C_p^{-1})]&=
		-\frac{1}{4}\,\frac{\lambda_{\rm p}^+}{\|\dev_3 \widetilde{\Sigma}\|}\, \tr[(\dev_3 \widetilde{\Sigma})^2]\\
		&=
		-\frac{1}{4}\,\frac{\lambda_{\rm p}^+}{\|\dev_3 \widetilde{\Sigma}\|}\,\| \dev_3( U_p^{-1}\,\widetilde{\Sigma}\, U_p)\|^2\leq 0,
	\end{align}
	\newpage
	We may summarize the properties  of the Lion 1997 model:
	\begin{itemize}
		\item[i)] from \eqref{choicchi00} it follows  that $C_p(t)\in{\rm Sym}(3)$;
		\item[ii)] plastic incompressibility: from  \eqref{choicchi00} together with $\det C_p(0)=1$ it follows that
		$\det C_p(t)=1$;
		\item[iii)] for the isotropic case, the right hand-side of \eqref{choicchi00} is a function of $C_p^{-1}$ and $C$ alone;
		\item[iii)] from i) and ii) together and using Lemma \ref{lemmaPSym} it follows that $C_p(t)\in{\rm PSym}(3)$;
		\item[v)] it is thermodynamically correct;
		\item[vi)] it is {an associated plasticity model} in the sense of Definition \ref{definitionpld}, see Proposition \ref{propositionegLion}.
	\end{itemize}
	
	\begin{remark}\label{remarkcompar}{\rm (Simo-Miehe 1992 model vs. Lion 1997 model)}
		In the anisotropic case, the flow rule proposed  by Simo and Miehe \cite{simo1993recent} (and later by Reese and Wriggers \cite{Reese97a} and Miehe \cite{Miehe92}) is not completely equivalent with the flow rule proposed by Lion (see also \cite{dettmer2004theoretical,brepols2014numerical}), since $\|\dev_3 \tau_e\|\neq \|\dev_3 \Sigma_{e} \|$ does not hold true in general. However, the difference is nearly absorbed by the positive plastic multipliers. The models may differ due to different yield conditions, but the flow rules are similar, having the same performance with respect to the thermodynamic consistency. Both models are  consistent according to our Definition \ref{consistentdef}, but we may not switch between them, since different elastic domains are considered, namely $\mathcal{E}_{{\Sigma}_e}$ and $\mathcal{E}_{{\tau}_e}$, respectively. This is in fact the main difference between these two models. Having different elastic domains we have different boundary points, since a point of the boundary of $\mathcal{E}_{{\tau}_e}$ is not necessarily on the boundary of $\mathcal{E}_{{\tau}_e}$. Hence, in these two flow rules we have a different behaviour corresponding to the indicator function of different domains. The material may reach the boundary of the elastic domain $\mathcal{E}_{{\tau}_e}$, while it is strictly inside the elastic domain $\mathcal{E}_{{\Sigma}_e}$, for the same local response.
	\end{remark}
	
	However, we have the following result:
	\begin{proposition}\label{eqlionsimo}
		In the isotropic case  the flow rule proposed  by Simo and Miehe \cite{simo1993recent} is  equivalent with the flow rule proposed by Lion \cite{lion1997physically}.
	\end{proposition}
	\begin{proof}
		We compare the flow rules \eqref{frsm300} and \eqref{choicchi00} and the proof is complete.
	\end{proof}

	\section{The Simo and Hughes  1998 plasticity formulation in terms of a plastic metric}\setcounter{equation}{0}

	The book \cite{Simo98b} has been edited years after the untimely  death of J.C. Simo. In this book also a finite strain plasticity model is proposed. However, this model has a  subtle fundamental deficiency which we aim to describe in the interest of the reader. The flow rule  considered in  \cite[page 310]{Simo98b} is
	\begin{equation}\label{simofr}
		\frac{\rm d }{{\rm d t}}[\overline{C}_p^{\,-1}]=-\frac{2}{3}\lambda^+_{\rm p}\,\tr(B_e)\,F^{-1}\frac{\dev_n\tau_e}{\|\dev_n\tau_e\|}F^{-T},\qquad \overline{C}_p=\frac{C_p}{\det C_p^{1/3}}\,,
	\end{equation}
	where $B_e=F_e F_e^T$,
	$\tau_e =2\, F_e \, D_{C_e}[W(C_e)]\, F_e^T=2\, B_e \, D_{B _e}[W(B_e)]$ is the elastic Kirchhoff stress tensor and $\lambda^+_{\rm p}\geq 0$ is the consistency parameter. If the plastic flow is isochoric then $\det F_p=\det C_p=1$. However, we must always have $\det \overline{C}_p=1=\det \overline{C}_p^{\,-1}$ by definition of $\overline{C}_p$. Since $F_e=F\, F_p^{-1}$, we have $\tr(B_e)=\langle F_p^{-T}F^T\,F\, F_p^{-1}\rangle=\langle \id, C\, C_p^{-1}\rangle=\tr(C\, C_p^{-1})$. Moreover, note that for elastically isotropic materials it holds
	\begin{align}\label{isocasesimo}
		D_{C_e}[W({C_e})]&=\alpha_1\, \id +\alpha_2\, C_e+\alpha_3 \, C_e^2\in{\rm Sym}(3),\notag\\ \tau_e &=2\, F_e \,[\alpha_1\, \id +\alpha_2\, C_e+\alpha_3 \, C_e^2]\, F_e^T,
	\end{align}
	where $\alpha_1,\alpha_2,\alpha_3$ are scalar functions of the invariants of $C_e$ which are functions of $C\, C_p^{-1}$, see Lemma \ref{lemaplasticn}.
	Since  $C_e=F_p^{-T}C\, F_p^{-1}$, we obtain
	\begin{align}\label{tauisocasesimo}
		\tau_e&= 2\, F\, F_p^{-1} \,[\alpha_1\, \id +\alpha_2\, F_p^{-T}C\, F_p^{-1}+\alpha_3 \, F_p^{-T}C\, F_p^{-1}F_p^{-T}C\, F_p^{-1}] F_p^{-T}\, F^T\notag\\&=2\, F \,f_1(C,C_p^{-1}) \, F^T,
	\end{align}
	with
	$
	f_1(C,C_p^{-1}):=\alpha_1\, C_p^{-1} +\alpha_2\, C_p^{-1} C\, C_p^{-1} +\alpha_3 \, C_p^{-1} C\, C_p^{-1} \,C\, C_p^{-1}\in {\rm Sym}(3).
	$
	Thus, for elastically isotropic materials  we deduce
	\begin{align}\label{taudevsimo}
		F^{-1}\,[{\dev_n\tau_e}]\,F^{-T}=&2\,f_1(C,C_p^{-1}) \notag-\frac{2}{3}\tr(F\, f_1(C,C_p^{-1}) \, F^T)\, C^{-1}\notag\\=&2\,f_1(C,C_p^{-1})-\frac{2}{3}\langle f_1(C,C_p^{-1}),C\rangle\, C^{-1}\in {\rm Sym}(3),\\
		\|\dev_3\tau_e\|=&\sqrt{\|\tau_e\|^2-\frac{1}{9}[\tr(\tau_e)]^2}\notag\\
		=&2\sqrt{\langle  f_1(C,C_p^{-1}) \cdot C,C\cdot \,f_1(C,C_p^{-1}) \rangle^2-\frac{1}{9}\langle f_1(C,C_p^{-1})  ,C \rangle^2}.\notag
	\end{align}
	Hence, $F^{-1}\frac{\dev_n\tau_e}{\|\dev_n\tau_e\|}F^{-T}\in {\rm Sym}(3)$  is a function of $C, {C}_p^{\,-1}$. Therefore the flow rule \eqref{simofr} can  entirely be expressed in terms of $C $ and $ {C}_p^{\,-1}$ alone.

	\begin{remark}
		It is not true, in general, that the right hand side of the flow rule \eqref{simofr} is in concordance with $\det \overline{C}_p(t)=1$, assuming that $\det \overline{C}_p(0)=1$.
	\end{remark}
	\begin{proof}
		From \eqref{simofr} we obtain by right multiplication with $\overline{C}_p$
		\begin{align}\label{simofr0}
			\frac{\rm d }{{\rm d t}}[\overline{C}_p^{\,-1}]\,\overline{C}_p=-\frac{2}{3}\,(\det C_p)^{-1/3}\,\lambda^+_{\rm p}\,\tr(C\,C_p^{-1})\,F_p^{-1}F_e^{-1}\frac{\dev_n\tau_e}{\|\dev_n\tau_e\|}F_e^{-T}F_p.
		\end{align}
		On the other hand, we have\footnote{Let us remark that $\frac{\rm d}{{\rm d t}}[\det \overline{C}_p^{\,-1}]=\det \overline{C}_p^{\,-1}\langle \id, \frac{\rm d}{{\rm d t}}[ \overline{C}_p^{\,-1}] \overline{C}_p\rangle$ shows that $\det \overline{C}_p^{\,-1}\!\!>0$ by direct integration of the ordinary differential equation.}
		\begin{align*}
			\frac{\rm d}{{\rm d t}}[\det \overline{C}_p^{\,-1}]=\langle {\rm Cof}\,\,\overline{C}_p^{\,-1}, \frac{\rm d}{{\rm d t}}[ \overline{C}_p^{\,-1}]\rangle=\det \overline{C}_P^{\,-1}\langle \overline{C}_p, \frac{\rm d}{{\rm d t}}[ \overline{C}_p^{\,-1}]\rangle
			=\det \overline{C}_P^{\,-1}\langle \id, \frac{\rm d}{{\rm d t}}[ \overline{C}_p^{\,-1}] \overline{C}_p\rangle.
		\end{align*}
		Hence, we deduce
		\begin{align}\label{simofr0}
			\frac{\rm d}{{\rm d t}}[\det \overline{C}_p^{\,-1}]=-\frac{2}{3}\,(\det C_p)^{-1/3}\,\lambda^+_{\rm p}\,\tr(C\,C_p^{-1})\langle\,F_e^{-1}
			\frac{\dev_n\tau_e}{\|\dev_n\tau_e\|}F_e^{-T},\id\rangle.
		\end{align}
		Since $F_e^{-1}
		\frac{\dev_n\tau_e}{\|\dev_n\tau_e\|}F_e^{-T}$ is not necessarily  a trace free matrix, we can not conclude that $\det \overline{C}_p^{\,-1}(t)=const\!.$ for all $t>0$. For instance, for elastically isotropic materials   (see \eqref{taudevsimo}) we
		have
		\begin{align}\label{taudevsimo2}
			{\dev_3\tau_e}
			=&2 \langle f_1(C,C_p^{-1}),C_p\rangle-\frac{2}{3}\langle f_1(C,C_p^{-1}),C\rangle\langle C^{-1},C_p \rangle
			\\
			=&2 \,\alpha_1\left[\langle C_p^{-1},C_p\rangle-\frac{1}{3}\langle C_p^{-1},C\rangle\langle C^{-1},C_p \rangle\right]\notag\\&+2 \,\alpha_2\left[\langle  C_p^{-1} C\, C_p^{-1},C_p\rangle-\frac{1}{3}\langle  C_p^{-1} C\, C_p^{-1},C\rangle\langle C^{-1},C_p \rangle\right]\notag\\&+2 \,\alpha_3\left[\langle C_p^{-1} C\, C_p^{-1} \,C\, C_p^{-1},C_p\rangle-\frac{1}{3}\langle C_p^{-1} C\, C_p^{-1} \,C\, C_p^{-1},C\rangle\langle C^{-1},C_p \rangle\right],\notag
		\end{align}
		which shows that $F_e^{-1}
		\frac{\dev_n\tau_e}{\|\dev_n\tau_e\|}F_e^{-T}$ is not necessarily a trace free matrix, see Appendix A.4.
	\end{proof}

	Summarizing the properties of the flow rule \eqref{simofr} we have:
	\begin{itemize}
		\item[i)] it is thermodynamically correct;
		\item[ii)] the right hand side is  a function of $C$ and $\overline{C}_p^{-1}$ only;
		\item[iii)] from this flow rule it follows $\overline{C}_p(t)\in {\rm Sym}(3)$ and $\det \overline{C}_p(t)>0$. Hence, it follows that $\overline{C}_p(t)\in {\rm PSym}(3)$;
		\item[iv)] plastic incompressibility: however, it does {\bf not follow} from the flow rule  that $\boldsymbol{\det \overline{C}_p(t)=1}$ (which must hold by the very definition of $\overline{C}_p$, since the right hand side is not trace-free, in general;
		\item[v)] it is {\bf not an associated plasticity model} in the sense of Definition \ref{definitionpld}.
	\end{itemize}

	\section{The Helm 2001 model}\label{helm}
	\setcounter{equation}{0}
	In this section we consider the model proposed by Helm \cite{helm2001formgedachtnislegierungen},  Vladimirov, Pietryga and Reese \cite[Eq. 25]{vladimirov2008modelling} (see also \cite{reese2008finite} and \cite[Eq. 55]{shutov2008finite} and the model considered by Brepols, Vladimirov and  Reese \cite[page 16]{brepols2014numerical}, Shutov and Ihlemann \cite[Eq. 80]{shutov2013analysis}). We prove later that this model is similar to the model considered by Miehe \cite{miehe1995theory} in 1995, provided certain interpretations are included.  Vladimirov, Pietryga and Reese \cite[Eq. 25]{vladimirov2008modelling} considered the following flow rule
	\begin{align}\label{SHflow}
		\frac{\rm d}{\rm dt}[ C_p]=\lambda_{\rm p}^+\,\frac{\dev_3 \widetilde{\Sigma}}{\sqrt{\tr( (\dev_3 \widetilde{\Sigma})^2)}}\cdot C_p\,,
	\end{align}
	where $\widetilde{\Sigma}=2\,C\, D_C[\widetilde{W}(C\, C_p^{-1})]$ is not necessarily symmetric for general $C_p\in{\rm PSym}(3)$, while $(\dev_3 \widetilde{\Sigma})\cdot C_p\in {\rm Sym}(3)$, see Section \ref{Miehe1995}. Therefore, we have
	$
	\dd\frac{\rm d}{\rm dt}[ C_p]\in {\rm Sym}(3).
	$
	The flow rule \eqref{SHflow} implies
	\begin{align}\label{ShuIhl}
		\frac{\rm d}{\rm dt} [\widetilde{W}(C\, C_p^{-1})]&=\langle D[\widetilde{W}(C\, C_p^{-1})],C\, \frac{\rm d}{\rm dt}[ C_p^{-1}]\rangle=
		\langle C\, D_C[\widetilde{W}(C\, C_p^{-1})]\, C_p, \frac{\rm d}{\rm dt}[ C_p^{-1}]\rangle\notag\\
		&=\frac{1}{2}
		\langle \widetilde{\Sigma}\, C_p, \frac{\rm d}{\rm dt}[ C_p^{-1}]\rangle=\frac{1}{2}
		\langle C_p^{-1}\,\widetilde{\Sigma}\, C_p, C_p \frac{\rm d}{\rm dt}[ C_p^{-1}]\rangle\\
		&=
		-\frac{1}{2}\frac{\lambda_{\rm p}^+\,}{\sqrt{\tr( (\dev_3 \widetilde{\Sigma})^2)}}\langle C_p^{-1}\,\widetilde{\Sigma}\, C_p, \dev_3 \widetilde{\Sigma}\rangle.\notag
	\end{align}
	Thus, using \eqref{ShuIhl00} and \eqref{ShuIhl220} we deduce
	\begin{align}\label{ShuIhl002}
		\frac{\rm d}{\rm dt} [\widetilde{W}(C\, C_p^{-1})]&=-\frac{1}{2}\frac{\lambda_{\rm p}^+}{\sqrt{\tr( (\dev_3 \widetilde{\Sigma})^2)}}\,\tr[(\dev_3 \widetilde{\Sigma})^2]\notag\\
		&=
		-\frac{\lambda_{\rm p}^+}{2}\| \dev_3( U_p^{-1}\,\widetilde{\Sigma}\, U_p)\|\leq 0,
	\end{align}
	which shows thermodynamical consistency.

	Summarizing, the Helm 2001  (Reese 2008 and Shutov-Ihlemann 2014) model has the following properties:
	\begin{itemize}
		\item[i)] from \eqref{ShuIhl002} it follows that it is {thermodynamically correct};
		\item[ii)] plastic incompressibility: from  \eqref{SHflow} and \eqref{derivaredet} it follows that
		$\det C_p(t)=1$;
		\item[iii)] for the isotropic case, the right hand-side of the flow rule \eqref{SHflow} is a function of $C_p^{-1}$ and $C$ alone;
		\item[iv)]  from $
		\frac{\rm d}{\rm dt}[ C_p]\in {\rm Sym}(3).
		$ it follows, in the isotropic case, that $C_p(t)\in{\rm Sym}(3)$;
		\item[v)] from ii) and iii) together and using Lemma \ref{lemmaPSym} it follows that $C_p(t)\in{\rm PSym}(3)$;
		\item[vi)] it has formally subdifferential structure, see Proposition \ref{prophelm}. However, the elastic domain
		$\mathcal{E}_{\rm e}(\widetilde{\Sigma},{\frac{2}{3}}\, \sigma_{\textbf{y}}^2)$  is not convex w.r.t $\widetilde{\Sigma}$, see Remark \ref{remarkconvexityMi}.
	\end{itemize}
	
	Moreover, we see that the following result holds:
	
	\begin{proposition}\label{prophelmmiehe}
		The flow rule  considered by Helm \cite{helm2001formgedachtnislegierungen}   coincides with the flow rule \eqref{SHflow}, i.e. with  the interpretation of Miehe's proposal \cite{Miehe95} presented in Section \ref{Miehe1995}.
	\end{proposition}
	\begin{proof} The proof follows from \eqref{SHflow} and combined  with \eqref{echimiehehelm}.
	\end{proof}

	\section{The Grandi-Stefanelli  2015 model}\setcounter{equation}{0}
	
	In this section we present  a model based on one representation used  by Grandi and Stefanelli \cite{GrandiStefanelli} and previously used by Frigeri and Stefanelli \cite[page 7]{frigeri2012existence}.
	We start by computing
	\begin{align}
		\frac{\rm d}{\rm dt} \widetilde{W}(C\, C_p^{-1})&=\langle D\widetilde{W}(C\,C_p^{-1}), C\, \frac{\rm d}{\rm dt} [C_p^{-1}]\rangle
		=\langle {\rm sym}[C\,D\widetilde{W}(C\,C_p^{-1})],  \frac{\rm d}{\rm dt} [C_p^{-1}]\rangle\notag
		\\
		&=\langle \underbrace{\sqrt{C_p}^{\,-1}\, {\rm sym}[C\,D\widetilde{W}(C\,C_p^{-1})]\,\sqrt{C_p}^{\,-1}}_{:=\frac{1}{2}\,\overset{\circ}{\Sigma}\,\in\, {\rm Sym}(3)}, \sqrt{C_p}\,\frac{\rm d}{\rm dt} [C_p^{-1}]\, \sqrt{C_p}\rangle.
	\end{align}
	It is now easy to see that, if we choose
	\begin{align}\label{Mielkeflow}
		\sqrt{C_p}\,\frac{\rm d}{\rm dt} [C_p^{-1}]\, \sqrt{C_p}\in -\Partial_{\overset{\circ}{\Sigma}} \Chi(\dev_3\overset{\circ}{\Sigma}),
	\end{align}
	where
	$ {\Chi}(\dev_3\overset{\circ}{\Sigma})$ is the indicator function of the convex elastic domain
	\begin{align*}
		\overset{\circ}{\mathcal{E}}_e(\overset{\circ}{\Sigma},{\frac{1}{3}}\, \sigma_{\textbf{y}}^2):=\left\{\overset{\circ}{\Sigma}\in {\rm Sym}(3)\, | \, \|\dev_3\overset{\circ}{\Sigma}\|^2\leq {\frac{1}{3}}\, \sigma_{\textbf{y}}^2\right\},
	\end{align*}
	then $C_p\in{\rm Sym}(3)$ and the reduced dissipation  inequality $\frac{\rm d}{\rm dt} \widetilde{W}(C\, C_p^{-1})\leq 0$ is satisfied. Thus, the model is thermodynamically correct.
	We also remark  that the flow rule \eqref{Mielkeflow} implies
	\begin{align*}
		\tr(\,\frac{\rm d}{\rm dt} [C_p^{-1}]\, {C_p})=\langle \frac{\rm d}{\rm dt} [C_p^{-1}] \sqrt{C_p}\,\sqrt{C_p},\id\rangle=\langle \sqrt{C_p}\,\frac{\rm d}{\rm dt} [C_p^{-1}] \sqrt{C_p},\id\rangle=0.
	\end{align*}
	Hence, we obtain $\det C_p(t)=1$ and further $C_p(t)\in{\rm PSym}(3)$.
	
	Using Lemma \ref{lemmasigmatild}, we give next some new representations of the stress-tensor
	\begin{align*}
		\overset{\circ}{\Sigma}:=2\,\sqrt{C_p}^{\,-1}\, {\rm sym}[C\,D\widetilde{W}(C\,C_p^{-1})]\,\sqrt{C_p}^{\,-1}=2\,{\rm sym}\left[\sqrt{C_p}^{\,-1}\, (C\,D\widetilde{W}(C\,C_p^{-1}))\,\sqrt{C_p}^{\,-1}\right]
	\end{align*}
	in terms of the stress tensors $\Sigma_e, \widetilde{\Sigma}$ and $\tau_e$, respectively. From \eqref{definitiesigma} we obtain
	$
	C\,D\widetilde{W}(C\,C_p^{-1})=\frac{1}{2}\,\widetilde{\Sigma}\, C_p
	$. We also use $F_p=R_p\, U_p=R_p\, \sqrt{C_p}$ and $F=F_eF_p$. Hence,  we deduce
	\begin{align*}
		\overset{\circ}{\Sigma}&={\rm sym}(\sqrt{C_p}^{\,-1}\, \widetilde{\Sigma}\,C_p\,\sqrt{C_p}^{\,-1})={\rm sym}(\sqrt{C_p}^{\,-1}\, \widetilde{\Sigma}\,\sqrt{C_p}),\\
		\overset{\circ}{\Sigma}&={\rm sym}(\sqrt{C_p}^{\,-1}\,F_p^T\,{\Sigma}_e\,F_p^{-T}\,\sqrt{C_p})={\rm sym}(\sqrt{C_p}^{\,-1}\,\sqrt{C_p}\,R_p^T\,{\Sigma}_e\,R_p\, \sqrt{C_p}^{-1}\,\sqrt{C_p})\\&={\rm sym}(R_p^T\,{\Sigma}_e\,R_p),\notag\\
		\overset{\circ}{\Sigma}&={\rm sym}(\sqrt{C_p}^{\,-1}\, F^T\,\tau_e\,F^{-T}\,\sqrt{C_p})={\rm sym}(\sqrt{C_p}^{\,-1}\, F_p^TF_e^T\,\tau_e\,F_e^{-T}F_p^{-T}\,\sqrt{C_p})\\&={\rm sym}(R_p^TF_e^T\,\tau_e\,F_e^{-T}R_p).\notag
	\end{align*}
	Note that $\Sigma_e$ is symmetric in case of elastic isotropy. Hence, for the isotropic case, we have
	\begin{align}\label{sesois}
		\overset{\circ}{\Sigma}&=R_p^T\,{\Sigma}_e\,R_p,\qquad \overset{\circ}{\Sigma}=R_p^T\, F_e^T\,{\tau}_e\, F_e^{-T}\,R_p .
	\end{align}
	However, we have
	$
	\|\overset{\circ}{\Sigma}\|^2=\|R_p^T\,{\Sigma}_e\,R_p\|^2=\|{\Sigma}_e\|^2
	,\
	\tr(\overset{\circ}{\Sigma})=\tr(R_p^T\,{\Sigma}_e\,R_p)=\tr(\Sigma_e).
	$
	Together, we obtain that
	$$
	\|\dev_3 \overset{\circ}{\Sigma}\|=\|\dev_3 {\Sigma}_e\|.
	$$
	
	In conclusion, for isotropic elastic materials we have the equivalence of the elastic domains
	\begin{align}\label{egalitate}
		\overset{\circ}{\mathcal{E}}_e(\overset{\circ}{\Sigma},\frac{1}{3}{\boldsymbol{\sigma}}_{\!\mathbf{y}}^2)=\mathcal{E}_{\rm e}(\Sigma_e,\frac{1}{3}{\boldsymbol{\sigma}}_{\!\mathbf{y}}^2).
	\end{align}
	
	\newpage
	Therefore, the flow rule \eqref{Mielkeflow} proposed by Grandi and Stefanelli \cite{GrandiStefanelli} has the following properties:
	\begin{itemize}
		\item[i)] it is thermodynamically correct;
		\item[ii)] from this flow rule it follows ${C}_p(t)\in {\rm Sym}(3)$ and $\det {C}_p(t)=1$. Hence, it follows that ${C}_p(t)\in {\rm PSym}(3)$;
		\item[iii)] the elastic domain $\overset{\circ}{\mathcal{E}}_e$ is convex w.r.t. $\overset{\circ}{\Sigma}$;
		\item[iv)] it is  { an associated plasticity model} in the sense of Definition \ref{definitionpld};
		\item[v)] it preserves ellipticity in elastic loading if the energy is elliptic throughout the domain $\overset{\circ}{\mathcal{E}}_e$ which makes it useful in association with the exponentiated Hencky energy $W_{_{\rm eH}}$ \cite{NeffGhibaLankeit,NeffGhibaPoly,GhibaNeffSilhavy,NeffGhibaPlasticity,NeffGhibaAdd,montella2015exponentiated,GhibaNeffMartin}.
	\end{itemize}
	
	We finish this section by comparing the Helm 2001 model and the Lion 1997 flow rule with the Grandi-Stefanelli 2015 model.

	\begin{proposition}\label{propositionegLion}
		In the isotropic case, the Lion 1997 flow rule (i.e. the Dettmer-Reese 2004 model \cite{dettmer2004theoretical}) is equivalent with the Grandi-Stefanelli 2015 flow rule.
	\end{proposition}
	\begin{proof}
		We recall that the flow rule of the  Lion 1997 model is
		\begin{align}\label{choicchi000}
			\frac{\rm d}{{\rm d t}}[C_p^{-1}]=-\lambda^+_{\rm p}\,F_p^{-1}\, \frac{\dev_3\Sigma_{e}}{\|\dev_3\Sigma_{e}\|}\,F_p^{-T}, \quad \lambda^+_{\rm p}\in \R_+, \
		\end{align}
		for ${\Sigma}_e\not\in {\rm int}(\mathcal{E}_{\rm e}({\Sigma_{e}},\frac{1}{3}{\boldsymbol{\sigma}}_{\!\mathbf{y}}^2))$.
		Since $F_p=R_p\, \sqrt{C_p}$ and  in the isotropic case ${\Sigma}_e=R_p\,\overset{\circ}{\Sigma}\,R_p^T$, using \eqref{egalitate} we rewrite the Lion's flow rule in the form
		\begin{align*}
			\frac{\rm d}{{\rm d t}}[C_p^{-1}]=-\lambda^+_{\rm p}\,\sqrt{C_p}^{-1}\,R_p^T \frac{\dev_3(R_p\,\overset{\circ}{\Sigma}\,R_p^T)}{\|\dev_3(R_p\,\overset{\circ}{\Sigma}\,R_p^T)\|}\,R_p\,\sqrt{C_p}^{-1}, \qquad \lambda^+_{\rm p}\in \R_+,\notag
		\end{align*}
		for $R_p\,\overset{\circ}{\Sigma}\,R_p^T\not\in {\rm int}(\overset{\circ}{\mathcal{E}}_e(\overset{\circ}{\Sigma},\frac{1}{3}{\boldsymbol{\sigma}}_{\!\mathbf{y}}^2))$,
		which is equivalent with
		\begin{align}\label{choicchi00sigma}
			\sqrt{C_p}\frac{\rm d}{{\rm d t}}[C_p^{-1}]\sqrt{C_p}=-\lambda^+_{\rm p}\, \frac{\dev_3\overset{\circ}{\Sigma}}{\|\dev_3\overset{\circ}{\Sigma}\|}, \qquad \lambda^+_{\rm p}\in \R_+,
		\end{align}
		for $R_p\,\overset{\circ}{\Sigma}\,R_p^T\not\in {\rm int}(\overset{\circ}{\mathcal{E}}_e(\overset{\circ}{\Sigma},\frac{1}{3}{\boldsymbol{\sigma}}_{\!\mathbf{y}}^2))$.
		Moreover, $R_p\,\overset{\circ}{\Sigma}\,R_p^T\in {\rm int}(\overset{\circ}{\mathcal{E}}_e(\overset{\circ}{\Sigma},\frac{1}{3}{\boldsymbol{\sigma}}_{\!\mathbf{y}}^2))\ \Leftrightarrow\
		\overset{\circ}{\Sigma}\in {\rm int}(\overset{\circ}{\mathcal{E}}_e(\overset{\circ}{\Sigma},\frac{1}{3}{\boldsymbol{\sigma}}_{\!\mathbf{y}}^2))$. Therefore, the flow rule \eqref{choicchi00sigma} becomes
		\begin{align}\label{choicchi000sigma}
			\sqrt{C_p}\frac{\rm d}{{\rm d t}}[C_p^{-1}]\sqrt{C_p}\in -\Partial_{\overset{\circ}{\Sigma}} \Chi(\dev_3\overset{\circ}{\Sigma}),
		\end{align}
		which coincides with the Grandi-Stefanelli 2015 flow rule \eqref{Mielkeflow}.
	\end{proof}
	
	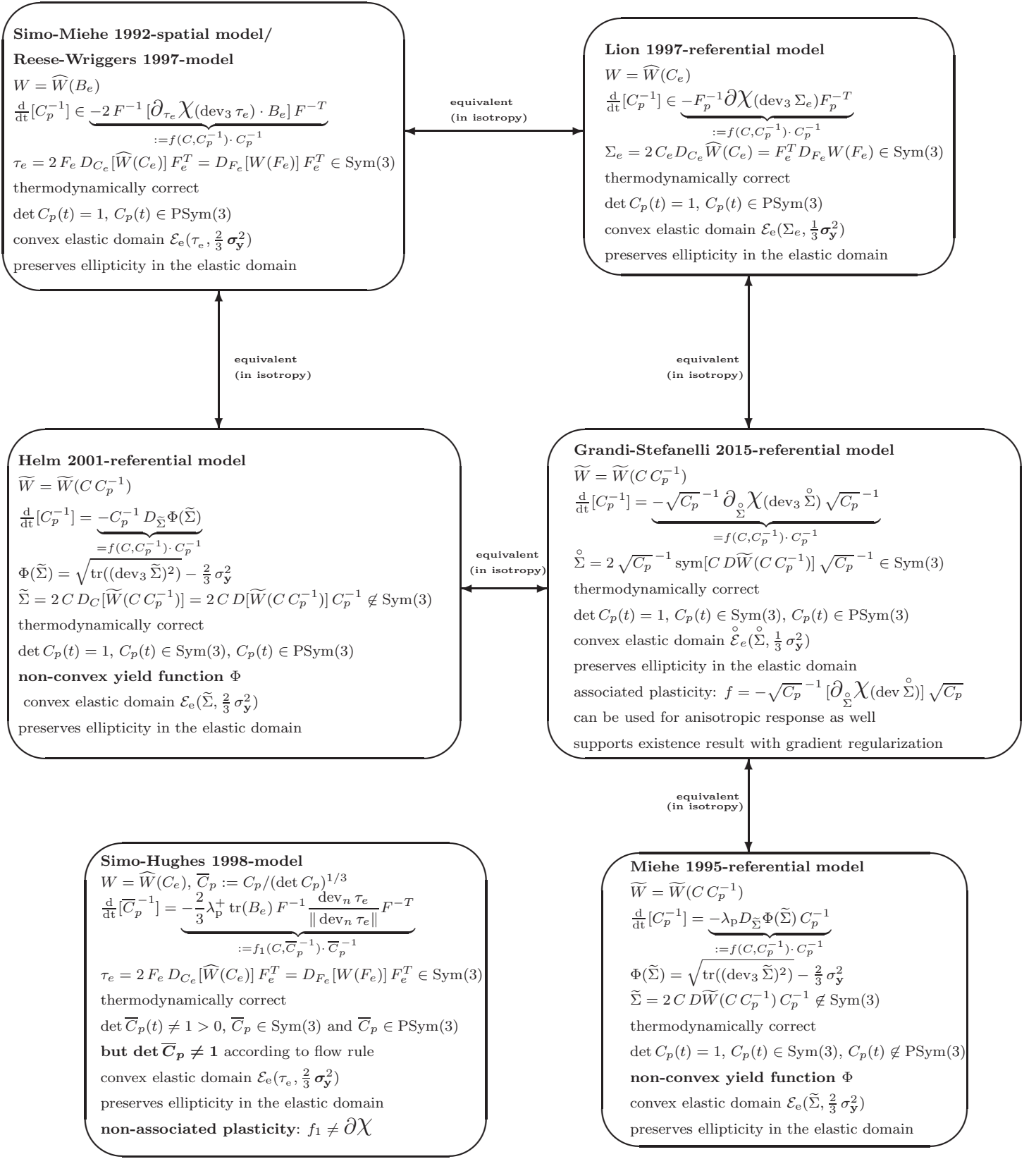
\begin{figure}
		\setlength{\unitlength}{0.97mm}
		\begin{center}
			\begin{picture}(10,205)
			\thicklines
			
			\put(52,192){\oval(72,50)}
			\put(20,210){\footnotesize{\bf Lion 1997-referential model}}
			\put(20,205){\footnotesize{${W}=\widehat{W}(C_e)$}}
			\put(20,200){\footnotesize{$\frac{\rm d}{{\rm d t}}[C_p^{-1}]\in\underbrace{- F_p^{-1}\Partial\Chi (\dev_3\Sigma_{e})F_p^{-T}}_{:=f(C,C_p^{-1})\cdot\, C_p^{-1}}$}}
			\put(20,190){\footnotesize{$\Sigma_{e}=2\,C_e D_{C_e} \widehat{W}(C_e)=F_e^T D_{F_e} {W}(F_e)\in{\rm Sym}(3)$}}
			\put(20,185){\footnotesize{thermodynamically correct}}
			\put(20,180){\footnotesize{$\det C_p(t)=1$, $C_p(t)\in{\rm PSym}(3)$}}
			\put(20,175){\footnotesize{convex elastic  domain $\mathcal{E}_{\rm e}({\Sigma_{e}},\frac{1}{3}{\boldsymbol{\sigma}}_{\!\mathbf{y}}^2)$}}
			\put(20,170){\footnotesize{preserves ellipticity in the elastic domain}}
			\put(20,170){\footnotesize{}}
			\put(23,165){\footnotesize{\bf }}
			\put(-58,192){\oval(78,56)}
			\put(-95,213){\footnotesize{\bf Simo-Miehe 1992-spatial model/}}
			\put(-95,208){\footnotesize{\bf Reese-Wriggers 1997-model}}
			\put(-95,203){\footnotesize{${W}=\widehat{W}(B_e)$}}
			\put(-95,198){\footnotesize{$\frac{\rm d}{\rm dt}[C_p^{-1}]\in \underbrace{-2 \, F^{-1}\,[\Partial_{\tau_e} \Chi(\dev_3\tau_e)\cdot B_e]\, F^{-T}}_{:={f}(C,C_p^{-1})\cdot\, C_p^{-1}}$}}
			\put(-95,188){\footnotesize{$\tau_e =2\, F_e \, D_{C_e}[\widehat{W}(C_e)]\, F_e^T=D_{F_e}[{W}(F_e)]\, F_e^T\in{\rm Sym}(3)$}}
			\put(-95,183){\footnotesize{thermodynamically correct}}
			\put(-95,178){\footnotesize{$\det C_p(t)=1$, $C_p(t)\in{\rm PSym}(3)$}}
			\put(-95,173){\footnotesize{convex elastic domain $\mathcal{E}_{\rm e}(\tau_{_{\rm e}},\frac{2}{3}\, {\boldsymbol{\sigma}}_{\!\mathbf{y}}^2)$}}
			\put(-95,168){\footnotesize{preserves ellipticity in the elastic domain}}
			\put(-95,170){\footnotesize{}}
			\put(-95,165){\footnotesize{\bf }}
			\put(-19,195){\vector(1,0){35}}
			\put(14,195){\vector(-1,0){33}}
			\put(-10,200){\tiny{\bf equivalent}}
			\put(-10,197){\tiny{\bf (in isotropy)}}

			\put(-52,105){\oval(88,64)}
			\put(-94,130){\footnotesize{\bf  Helm 2001-referential model}}
			\put(-94,125){\footnotesize{$\widetilde{W}=\widetilde{W}(C\, C_p^{-1})$}}
			\put(-94,119){\footnotesize{$\frac{\rm d}{\rm dt}[ C_p^{-1}]=\underbrace{-C_p^{-1}\,D_{\widetilde{\Sigma}} \Phi(\widetilde{\Sigma})}_{=f(C,C^{-1}_p)\cdot \,C_p^{-1}}$}}
			\put(-94,108){\footnotesize{$\Phi(\widetilde{\Sigma})=\sqrt{\tr( (\dev_3 \widetilde{\Sigma})^2)}-{\frac{2}{3}}\, \sigma_{\textbf{y}}^2$}}
			\put(-94,103){\footnotesize{$\widetilde{\Sigma}=2\,C\, D_C[\widetilde{W}(C\, C_p^{-1})]=2\,C\, D [\widetilde{W}(C\, C_p^{-1})]\, C_p^{-1}\not\in{\rm Sym}(3)$}}
			\put(-94,98){\footnotesize{thermodynamically correct}}
			\put(-94,93){\footnotesize{$\det C_p(t)=1$, ${C_p(t)\in{\rm Sym}(3)}$, ${C_p(t)\in{\rm PSym}(3)}$}}
			\put(-94,88){\footnotesize{{\bf non-convex  yield function $\Phi$}}}
			\put(-94,83){\footnotesize{{  convex elastic domain $\mathcal{E}_{\rm e}(\widetilde{\Sigma},{\frac{2}{3}}\, \sigma_{\textbf{y}}^2)$}}}
			\put(-94,78){\footnotesize{preserves ellipticity in the elastic domain}}
			\put(-94,63){\footnotesize{}}
			
			\put(-6,106){\vector(1,0){15}}
			\put(5,106){\vector(-1,0){13}}
			\put(-5,112){\bf \tiny{equivalent}}
			\put(-6.3,109){\bf\tiny{(in isotropy)}}
			
			\put(-55,137){\vector(0,1){27}}
			\put(-55,142){\vector(0,-1){5}}
			\put(-52,150){\bf \tiny{equivalent}}
			\put(-52,147){\bf \tiny{(in isotropy)}}
			
			\put(48,138){\vector(0,1){29}}
			\put(48,164){\vector(0,-1){27}}
			\put(34,150){\bf \tiny{equivalent}}
			\put(32,147){\bf \tiny{(in isotropy)}}
			
			\put(55,105){\oval(92,64)}
			\put(14,132){\footnotesize{\bf  Grandi-Stefanelli 2015-referential model}}
			\put(14,127){\footnotesize{$\widetilde{W}=\widetilde{W}(C\, C_p^{-1})$}}
			\put(14,122){\footnotesize{$\frac{\rm d}{\rm dt}[ C_p^{-1}]=\underbrace{- \sqrt{C_p}^{\, -1}\,\Partial_{\overset{\circ}{\Sigma}} \Chi(\dev_3\overset{\circ}{\Sigma})\,\sqrt{C_p}^{\, -1}}_{=f(C,C^{-1}_p)\cdot\, C_p^{-1}}$}}
			\put(14,110){\footnotesize{$\overset{\circ}{\Sigma}=2\,\sqrt{C_p}^{\,-1}\,{\rm sym}[C\, D\widetilde{W}(C\, C_p^{-1})]\,\sqrt{C_p}^{\,-1}\in{\rm Sym}(3)$}}
			\put(14,105){\footnotesize{thermodynamically correct}}
			\put(14,100){\footnotesize{$\det C_p(t)=1$, ${C_p(t)\in{\rm Sym}(3)}$, ${C_p(t)\in{\rm PSym}(3)}$}}
			\put(14,95){\footnotesize{convex elastic domain $\overset{\circ}{\mathcal{E}}_e(\overset{\circ}{\Sigma},{\frac{1}{3}}\, \sigma_{\textbf{y}}^2)$}}
			\put(14,90){\footnotesize{preserves ellipticity in the elastic domain}}
			\put(14,85){\footnotesize{associated plasticity:  $f= - \sqrt{C_p}^{\, -1}\,[\Partial_{\overset{\circ}{\Sigma}} \Chi(\dev \overset{\circ}{\Sigma})]\,\sqrt{C_p}$}}
			\put(14,80){\footnotesize{can be used for anisotropic response as well}}
			\put(14,75){\footnotesize{supports existence result with gradient regularization}}
			
			\put(-42,26){\oval(78,61)}
			\put(-78,52){\footnotesize{\bf Simo-Hughes 1998-model }}
			\put(-78,48){\footnotesize{$W=\widehat{W}(C_e)$,\  $\overline{C}_p:=C_p/(\det C_p)^{1/3}$}}
			\put(-78,43){\footnotesize{$\frac{\rm d }{{\rm d t}}[\overline{C}_p^{\,-1}]=\underbrace{-\frac{2}{3}\lambda^+_{\rm p}\,\tr(B_e)\,F^{-1}\frac{\dev_n\tau_e}{\|\dev_n\tau_e\|}F^{-T}}_{:=f_1(C,\overline{C}_p^{\,-1})\cdot\, \overline{C}_p^{\,-1}}$}}
			\put(-78,30){\footnotesize{$\tau_e =2\, F_e \, D_{C_e}[\widehat{W}(C_e)]\, F_e^T=D_{F_e}[{W}(F_e)]\, F_e^T\in{\rm Sym}(3)$}}
			\put(-78,25){\footnotesize{thermodynamically correct}}
			\put(-78,20){\footnotesize{$\det \overline{C}_p(t)\neq 1>0$,  $\overline{C}_p\in{\rm Sym}(3)$ and ${\overline{C}_p\in{\rm PSym}(3)}$ }}
			\put(-78,15){\footnotesize{{\bf but} $\boldsymbol{\det \overline{C}_p\neq 1}$ according to flow rule}}
			\put(-78,10){\footnotesize{convex elastic domain $\mathcal{E}_{\rm e}(\tau_{_{\rm e}},\frac{2}{3}\, {\boldsymbol{\sigma}}_{\!\mathbf{y}}^2)$}}
			\put(-78,5){\footnotesize{preserves ellipticity in the elastic domain}}
			\put(-78,0){\footnotesize{\bf  non-associated plasticity}: ${f_1\neq \Partial \Chi}$}
			
			\put(48,54.5){\vector(0,1){18.5}}
			\put(48,72){\vector(0,-1){17.5}}
			\put(34,65){\bf \tiny{equivalent}}
			\put(32,63){\bf \tiny{(in isotropy)}}
			
			\put(55,26){\oval(71,57)}
			\put(25,51){\footnotesize{\bf Miehe 1995-referential model}}
			\put(25,46){\footnotesize{$\widetilde{W}=\widetilde{W}(C\, C_p^{-1})$}}
			\put(25,41){\footnotesize{$\frac{\rm d}{\rm dt} [C_p^{-1}]=\underbrace{-\lambda_{\rm p}D_{\widetilde{\Sigma}}\Phi(\widetilde{\Sigma})\, C_p^{-1}}_{:={f}(C,C_p^{-1})\cdot\, C_p^{-1}}$}}
			\put(25,30){\footnotesize{$\Phi(\widetilde{\Sigma})=\sqrt{\tr( (\dev_3 \widetilde{\Sigma})^2)}-{\frac{2}{3}}\, \sigma_{\textbf{y}}^2$}}
			\put(25,25){\footnotesize{$\widetilde{\Sigma}= 2\,C\,D\widetilde{W}(C\,C_p^{-1})\, C_p^{-1}\not\in{\rm Sym}(3)$}}
			\put(25,20){\footnotesize{thermodynamically correct}}
			\put(25,15){\footnotesize{$\det C_p(t)=1$,  ${C_p(t)\in{\rm Sym}(3)}$, ${C_p(t)\not\in{\rm PSym}(3)}$}}
			\put(25,10){\footnotesize{\bf non-convex yield function $\Phi$}}
			\put(25,5){\footnotesize{convex elastic domain $\mathcal{E}_{\rm e}(\widetilde{\Sigma},{\frac{2}{3}}\, \sigma_{\textbf{y}}^2)$}}
			\put(25,0){\footnotesize{preserves ellipticity in the elastic domain}}
			\put(25,0){\footnotesize{}}
			\end{picture}
		\end{center}
		\caption{Idealized, isotropic perfect plasticity models involving a 6-dimensional flow rule for $C_p$   w.r.t. the reference configuration are considered. By definition, the trajectory for the plastic metric $C_p(t)$ should remain in ${\rm PSym}(3)$. $\lambda_{\rm p}^+$ is the plastic multiplier. We have recast all flow rules in the format $\frac{\rm d}{\rm dt}[P^{-1}]\, P\in  -{\PartialCaption\ChiCaption}$ or $ \sqrt{P}\,\frac{\rm d}{\rm dt}[P^{-1}]\, \sqrt{P}\in  -{\PartialCaption\ChiCaption} $.}\label{plastmodeldiagram}
	\end{figure}

	\begin{proposition}\label{prophelm}
		In the isotropic case, the Helm  2001 flow rule  is equivalent with the Grandi-Stefanelli 2015 flow rule, i.e. it is also equivalent with the Lion 1997 flow rule and the Dettmer-Reese 2004 model.
	\end{proposition}
	\begin{proof}
		We  have
		\begin{align*}
			\overset{\circ}{\Sigma}&={\rm sym}(\sqrt{C_p}^{\,-1}\, \widetilde{\Sigma}\,\sqrt{C_p})={\rm sym}(\sqrt{C_p}^{\,-1}\, \widetilde{\Sigma}\,C_p\, C_p^{-1}\,\sqrt{C_p})\notag\\&={\rm sym}(\sqrt{C_p}^{\,-1}\, (\widetilde{\Sigma}\,C_p)\sqrt{C_p}^{\,-1}),
		\end{align*}
		and we recall that for isotropic materials  $$\widetilde{\Sigma}\,C_p=F_p^{T} \Sigma_e\,F_p^{-T}\,C_p= F_p^{T} \Sigma_e\,F_p\in {\rm Sym}(3)$$ holds. Hence, for isotropic materials
		\begin{align*}
			\overset{\circ}{\Sigma}&=\sqrt{C_p}^{\,-1}\, (\widetilde{\Sigma}\,C_p)\sqrt{C_p}^{\,-1}=\sqrt{C_p}^{\,-1}\, \widetilde{\Sigma}\,\sqrt{C_p}, \qquad \, \widetilde{\Sigma}=\sqrt{C_p}\,\overset{\circ}{\Sigma}\, \sqrt{C_p}^{\,-1}\,.
		\end{align*}
		Using the above identity, we may rewrite the Helm 2001-flow rule \eqref{SHflow} in the form
		\begin{align}\label{SHflows0}
			\frac{\rm d}{\rm dt}[ C_p]=\lambda_{\rm p}^+\,\frac{\dev_3 (\sqrt{C_p}\,\overset{\circ}{\Sigma}\, \sqrt{C_p}^{\,-1})}{\sqrt{\tr( (\dev_3 (\sqrt{C_p}\,\overset{\circ}{\Sigma}\, \sqrt{C_p}^{\,-1}))^2)}}\cdot C_p\,.
		\end{align}
		We also have
		\begin{align*}
			\tr(\sqrt{C_p}\,\overset{\circ}{\Sigma}\, \sqrt{C_p}^{\,-1})&=\tr(\overset{\circ}{\Sigma}),\notag\\ \dev_3 (\sqrt{C_p}\,\overset{\circ}{\Sigma}\, \sqrt{C_p}^{\,-1})&=\sqrt{C_p}\,(\dev_3 \overset{\circ}{\Sigma})\, \sqrt{C_p}^{\,-1},\\
			\tr( [\dev_3 (\sqrt{C_p}\,\overset{\circ}{\Sigma}\, \sqrt{C_p}^{\,-1})]^2)&=\tr( [\sqrt{C_p}\,(\dev_3 \overset{\circ}{\Sigma})\, \sqrt{C_p}^{\,-1}]^2)=
			\tr( \sqrt{C_p}\,(\dev_3 \overset{\circ}{\Sigma})^2\, \sqrt{C_p}^{\,-1})\notag\\&=\tr((\dev_3 \overset{\circ}{\Sigma})^2)=\langle(\dev_3 \overset{\circ}{\Sigma})^2,\id\rangle\notag\\&=\langle\dev_3 \overset{\circ}{\Sigma},\dev_3 \overset{\circ}{\Sigma}\rangle=\|\dev_3 \overset{\circ}{\Sigma}\|^2.\notag
		\end{align*}
		Hence, Helm's flow rule \eqref{SHflow} is equivalent with
		\begin{align}\label{SHflows0}
			\frac{\rm d}{\rm dt}[ C_p]&=\lambda_{\rm p}^+\,\sqrt{C_p}\,\frac{\dev_3 \overset{\circ}{\Sigma}}{\|\dev_3 \overset{\circ}{\Sigma}\|}\, \sqrt{C_p}\,\notag\\&
			\qquad \Leftrightarrow\qquad
			\sqrt{C_p}\,\frac{\rm d}{\rm dt}[ C_p^{-1}]\, \sqrt{C_p}=-\lambda_{\rm p}^+\,\frac{\dev_3 \overset{\circ}{\Sigma}}{\|\dev_3 \overset{\circ}{\Sigma}\|}\,,
		\end{align}
		and the proof is complete.
	\end{proof}
	\begin{remark}
		The equivalence is true for an isotropic formulation only. However, the Grandi-Stefanelli model will provide a consistent flow-rule for a plastic metric also in the anisotropic case.
	\end{remark}

	An existence proof for the energetic formulation \cite{frigeri2012existence} of the model given by Grandi and Stefanelli \cite{GrandiStefanelli} together  with a full plastic strain regularization can be given along the lines of Mielke's energetic approach \cite{mielke2003energetic,mielke2004SIAM,mainik2005existence,mielke2006ZAMM,mainik2009global}.
	
	\section{Summary}
	\begin{figure}
		\setlength{\unitlength}{0.97mm}
		\begin{center}
			\begin{picture}(10,45)
			\thicklines
			\put(55,28){\oval(73,43)}
			\put(22,45){\footnotesize{\bf 9-dimensional flow rule}}
			\put(22,40){\footnotesize{$W=\widehat{W}(C_e)=\widetilde{W}(F\, F_p^{-1})$}}
			\put(22,35){\footnotesize{$-F_p\,\frac{\rm d}{{\rm d t}}[F_p^{-1}]\in\Partial \Chi(\dev_3 \Sigma_{e})$}}
			\put(22,30){\footnotesize{$\Sigma_{e}=2\,C_e D_{C_e} \widehat{W}(C_e)=F_e^T D_{F_e} {W}(F_e)\in{\rm Sym}(3)$}}
			\put(22,25){\footnotesize{thermodynamically correct, $\det F_p=1$}}
			\put(22,20){\footnotesize{convex elastic domain $\mathcal{E}_{\rm e}({\Sigma_{e}},\frac{1}{3}{\boldsymbol{\sigma}}_{\!\mathbf{y}}^2)$}}
			\put(22,15){\footnotesize{preserves ellipticity in the elastic domain}}
			\put(22,10){\footnotesize{associated plasticity: $f= \Partial \Chi$}}
			
			\put(-13,25){\vector(1,0){32}}
			\put(15,25){\vector(-1,0){32}}
			\put(-4,29){\bf \tiny{similar}}
			\put(-5,27){\bf \tiny{structure}}
			
			\put(-50,26){\oval(67,59)}
			\put(-80,51){\footnotesize{\bf Another referential model}}
			\put(-80,46){\footnotesize{$\widetilde{W}=\widetilde{W}(C\, C_p^{-1})$}}
			\put(-80,41){\footnotesize{$\frac{\rm d}{\rm dt} [C_p^{-1}]\in\underbrace{-\,\Partial_{\widetilde{\Sigma}}\widetilde{\Chi}(\dev_3\widetilde{\Sigma})\, C_p^{-1}}_{:={f}_2(C,C_p^{-1})\cdot\, C_p^{-1}}$}}
			\put(-80,30){\footnotesize{$\widetilde{\Sigma}= 2\,C\,D\widetilde{W}(C\,C_p^{-1})\, C_p^{-1}\not\in{\rm Sym}(3)$}}
			\put(-80,25){\footnotesize{thermodynamically correct}}
			\put(-80,20){\footnotesize{$\det C_p(t)=1$,  $\boldsymbol{C_p(t)\not\in{\rm Sym}(3)}$}}
			\put(-80,15){\footnotesize{$\boldsymbol{C_p(t)\not\in{\rm PSym}(3)}$}}
			\put(-80,10){\footnotesize{convex elastic domain $\mathcal{E}_{\rm e}(\widetilde{\Sigma},{\frac{2}{3}}\, \sigma_{\textbf{y}}^2)$}}
			\put(-80,5){\footnotesize{preserves ellipticity in the elastic domain}}
			\put(-80,0){\footnotesize{associated plasticity: $f_2= \Partial \widetilde{\Chi}$}}
			\end{picture}
		\end{center}
		\caption{An inconsistent model and a 9-dimensional flow rule for $F_p$. They are associative, since  both flow rules are in the format $ \frac{\rm d}{\rm dt}[P]\, P^{-1}\in  -{\PartialCaption\ChiCaption}$ or $ \frac{\rm d}{\rm dt}[\varepsilon_p]\in {\PartialCaption\ChiCaption}$.}\label{plastmodeldiagram2}
	\end{figure}
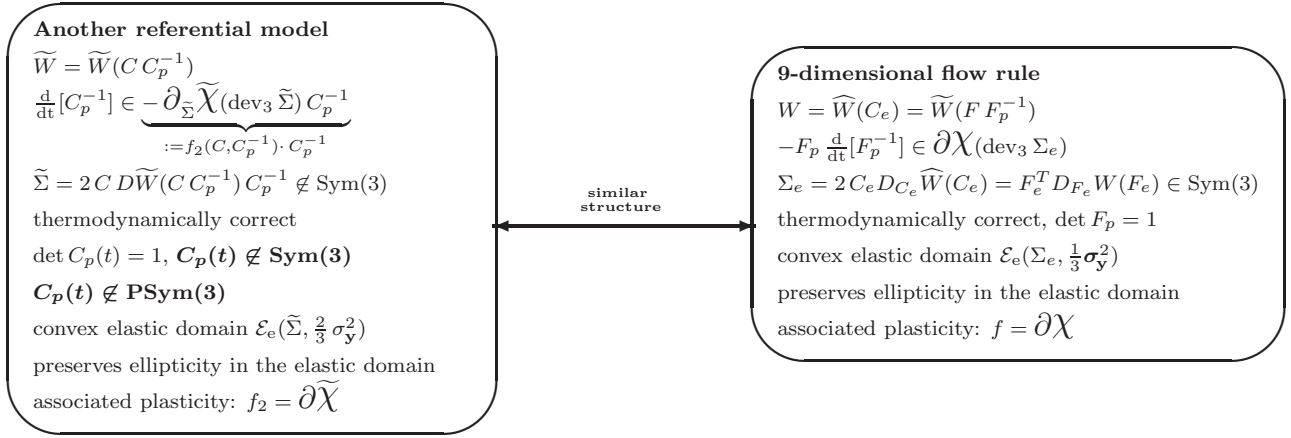
	In isotropic elasto-plasticity it is common knowledge that a reduction to a \textit{6-dimensional flow rule} for a \textit{plastic metric} $C_p$ is in principle possible. We have discussed several existing different models. Not all of them are free of inconsistencies.  This testifies to the fact that setting up a consistent 6-dimensional flow-rule  is not entirely trivial.
	
	One problem which often occurs, is that the flow rule for $C_p$ is written in terms of $F_p$, which however should not appear at all. One finding of our investigation is that, nevertheless, in the isotropic case, all consistent  flow rules can be expressed in $C_p$ alone and are equivalent. The Grandi-Stefanelli model \cite{GrandiStefanelli} has the decisive advantage to be operable also in the anisotropic case.
	In Figure \ref{plastmodeldiagram} and Figure \ref{plastmodeldiagram2} we summarize the investigated isotropic plasticity models and we indicate if the known conditions which make them  consistent are satisfied.

	\section{Acknowledgement}
		We would like to thank Prof. Stefanie Reese (RWTH\! Aachen), Prof. J\"{o}rn Ihlemann (TU\! Chemnitz), Dr.\! Alexey Shutov (TU\! Chemnitz) and Prof. Alexander Mielke (WIAS-Berlin) for in-depth discussion of flow rules in $C_p$.  I.D. Ghiba acknowledges support from the Romanian National Authority for Scientific Research (CNCS-UEFISCDI), Project No. PN-II-ID-PCE-2011-3-0521.

	\bibliographystyle{plain} %plain
	
		\appendix 
		\section*{Appendix}\addcontentsline{toc}{section}{Appendix} \addtocontents{toc}{\protect\setcounter{tocdepth}{-1}}
		
		\setcounter{section}{1}
		\subsection*{A.1 The Lion 1997 model for the Neo-Hooke elastic  energy}\setcounter{equation}{0}

		For a quick consistency check we exhibit the consistency of this model directly for a Neo-Hooke elastic energy  and  we  give the concrete expression for the functions $f(C,C_p^{-1})$ and $\widehat{f}(C,C_p^{-1})$. To this end, we consider the energy
		\begin{align*}
		\widehat{W}_{_{\rm NH}}(C_e)&=\mu\, \tr\left(\frac{C_e}{\det C_e^{1/3}}\right)+h(\det C)\overset{\det C_p=1}{=}\mu\, \tr\left(\frac{C_e}{\det C^{1/3}}\right)+h(\det C).\notag
		\end{align*}
		We deduce
		$
		\Sigma_{e}:=2\,C_e\, D_{C_e}[\widehat{W}({C_e})]=2\, C_e \,\mu\, \frac{1}{\det C^{1/3}}\cdot \id=2\,\mu\, \frac{1}{\det C^{1/3}} \cdot C_e.
		$
		Hence, the flow rule \eqref{choicchi00} can be written in the form
		
		\begin{align}\label{choicchiNH}
		\frac{\rm d}{{\rm d t}}[C_p^{-1}]&=-\lambda^+_{\rm p}\,F_p^{-1}\,\frac{\dev C_e}{\|\dev C_e\|}\,F_p^{-T}\notag\\&=-\frac{\lambda^+_{\rm p}}{\|\dev C_e\|} \,\left(F_p^{-1} C_e\,F_p^{-T}-\frac{1}{3}\,\tr(C_e)\cdot F_p^{-1} F_p^{-T}\right)\notag\\
		&=-\frac{\lambda^+_{\rm p}}{\|\dev C_e\|} \,\left(C_p^{-1}\, C\, C_p^{-1}-\frac{1}{3}\,\tr(C_p^{-1}\, C)\cdot C_p^{-1} \right).
		\end{align}
		
		We also deduce
		\begin{align*}
		\|\dev C_e\|^2&=\|C_e\|^2-\frac{1}{3}\, [\tr(C_e)]^2=\|F_p^{-T}\, C\, F_p^{-1}\|^2-\frac{1}{3}\, [\tr(F_p^{-T}\, C\, F_p^{-1})]^2\notag\\&=\langle F_p^{-T}\, C\, F_p^{-1},F_p^{-T}\, C\, F_p^{-1}\rangle-\frac{1}{3}\, \langle F_p^{-T}\, C\, F_p^{-1}),\id\rangle^2\notag\\
		&=\langle C_p^{-1}\, C, C\, C_p^{-1}\rangle-\frac{1}{3}\,[ \tr( C_p^{-1}\, C)]^2=[\tr (C_p^{-1}\, C)^2]-\frac{1}{3}\, [\tr( C_p^{-1}\, C)]^2.
		\end{align*}
		Therefore, we obtain
		\begin{align}\label{choicchiNH1}
		\frac{\rm d}{{\rm d t}}[C_p^{-1}]&
		=-\frac{\lambda^+_{\rm p}}{\sqrt{\tr [(C_p^{-1}\, C)^2]-\frac{1}{3}\, [\tr( C_p^{-1}\, C)]^2}}\,\left(C_p^{-1}\, C\, -\frac{1}{3}\,\tr(C_p^{-1}\, C)\cdot \id \right)\, C_p^{-1}\\
		&
		=-\frac{\lambda^+_{\rm p}}{\sqrt{[\tr (C_p^{-1}\, C)^2]-\frac{1}{3}\, \tr[( C_p^{-1}\, C)]^2}}\,\,C_p^{-1}\,\left( C\, C_p^{-1} -\frac{1}{3}\,\tr( C\, C_p^{-1})\cdot \id \right)\notag.
		\end{align}
		Comparing \eqref{fhat1},  \eqref{choicchiNH1}, \eqref{choicchi0f} and \eqref{explf}, we deduce
		\begin{align*}
		\widehat{f}(C,C_p^{-1})&=C,\notag\\
		{f}(C,C_p^{-1})
		&=\left\{\frac{-\lambda^+_{\rm p}}
		{\sqrt{\tr[(C\,C_p^{-1})^2]-\frac{1}{3}[\tr(C\,C_p^{-1})]^2}}\,
		\dev_3(\,C_p^{-1}
		\, C)\quad \big| \quad \lambda^+_{\rm p}\in \R_+\right\}.\notag
		\end{align*}
		We clearly see that even for this simple energy, we have $\tr[(C\,C_p^{-1})^2]-\frac{1}{3}[\tr(C\,C_p^{-1})]^2\neq \|\dev_3(C\,C_p^{-1})\|^2$,
		since if we assume the contrary we deduce
		\begin{align}\label{trdevpm}
		\tr[(C\,C_p^{-1})^2]&-\frac{1}{3}[\tr(C\,C_p^{-1})]^2= \|C\,C_p^{-1}\|^2-\frac{1}{3}[\tr(C\,C_p^{-1})]^2\notag\\&
		\Leftrightarrow \langle C\,C_p^{-1},(C\,C_p^{-1})^T\rangle= \pm \langle C\,C_p^{-1},C\,C_p^{-1}\rangle.
		\end{align}
		On the other hand, we deduce
		\begin{align}\label{ShuIhl000}
		\tr[(C\,C_p^{-1})^2]&=\langle(C\,C_p^{-1})\,  (C\,C_p^{-1}),\id\rangle=\langle C_p^{-1}\,C\,C_p^{-1}\,  (C\,C_p^{-1})\,C_p,\id\rangle\notag\\&=\langle C_p^{-1}\,C\,C_p^{-1}\, C_p\, (C\,C_p^{-1})^T,\id\rangle=\langle C_p^{-1}\,C,  C\,C_p^{-1}\rangle.
		\end{align}
		Since from Remark \ref{remarkcpPsym} it follows that $C_p\in{\rm PSym}(3)$, we further deduce that
		\begin{align}\label{ShuIhl2200}
		\langle C_p^{-1}\,C\,C_p^{-1}\, C_p, C\,C_p^{-1}\rangle&=
		\langle U_p^{-1}U_p^{-1}\,C\,C_p^{-1}\, U_p\,U_p,  C\,C_p^{-1}\rangle\notag\\&
		=
		\langle U_p^{-1}\,C\,C_p^{-1}\, U_p, U_p^{-1} C\,C_p^{-1}\,U_p\rangle\notag\\&=
		\| U_p^{-1}\,C\,U_p^{-1}\|^2\geq 0,
		\end{align}
		where $U_p^2=C_p$. Hence, $[\tr(C\,C_p^{-1})^2]\geq 0$ and from \eqref{trdevpm} we deduce
		\begin{align}\label{trdevpm0}
		\langle C\,C_p^{-1},(C\,C_p^{-1})^T\rangle&=  \langle C\,C_p^{-1},C\,C_p^{-1}\rangle\quad
		\Leftrightarrow \quad
		\langle C\,C_p^{-1},{\rm skew}(C\,C_p^{-1})\rangle= 0\notag\\&\Leftrightarrow \quad
		{\rm skew}(C\,C_p^{-1})= 0\quad \Leftrightarrow \quad
		C\,C_p^{-1}\in{\rm Sym}(3).
		\end{align}
		In conclusion, $\tr[(C\,C_p^{-1})^2]-\frac{1}{3}[\tr(C\,C_p^{-1})]^2\neq \|\dev_3(C\,C_p^{-1})\|^2$ and the flow-rule does not have  a subdifferential structure of the form $C_p\,\frac{\rm d}{{\rm d t}}[C_p^{-1}]\in -\,\Partial \Chi(\dev \Sigma)$.

		\subsection*{A.2 The Helm 2001 model for the Neo-Hooke energy}

		For  the simplest Neo-Hooke elastic energy  $W(F_e)=\tr(C_e)=\widetilde{W}(C\, C_p^{-1})=\frac{1}{2}\,\tr(C\, C_p^{-1})$, we have
		\begin{align}\label{NHSHI}
		D_C[\widetilde{W}(C\, C_p^{-1})]=\frac{1}{2}\,C_p^{-1} \qquad \Rightarrow\qquad \widetilde{\Sigma}=C\, C_p^{-1}\not\in{\rm Sym}(3),
		\end{align}
		and the flow rule  \eqref{SHflow} implies
		\begin{align*}
		\frac{\rm d}{\rm dt}[C_p]&=\lambda_{\rm p}^+\,\frac{\dev_3 (C\, C_p^{-1})}{\sqrt{\tr[ (\dev_3 (C\, C_p^{-1}))^2]}}\cdot C_p\notag\\&=\frac{\lambda_{\rm p}^+}{\sqrt{\tr[ (\dev_3 (C\, C_p^{-1}))^2]}}\,[\,C -\frac{1}{3}\tr(C\, C_p^{-1})\cdot C_p]\in {\rm Sym}(3)\quad \Rightarrow \quad  C_p\in {\rm Sym}(3),\notag
		\end{align*}
		and also
		\begin{align}\label{SHdet1}
		\frac{\rm d}{\rm dt}[C_p]\, C_p^{-1}=\lambda_{\rm p}^+\,\frac{\dev_3 (C\, C_p^{-1})}{\sqrt{\tr[ (\dev_3 (C\, C_p^{-1}))^2]}}\quad \Rightarrow \quad \det C_p=1.
		\end{align}
		The thermodynamical consistency may follow  from \eqref{ShuIhl002}. An alternative proof, directly for the Neo-Hooke case, results from  \eqref{ShuIhl} and \eqref{NHSHI}, since we have at fixed in time $C$
		\begin{align}\label{dissSH0}
		&\frac{\rm d}{\rm dt} [\widetilde{W}(C\, C_p^{-1})]=
		-\frac{\lambda_{\rm p}^+}{4\,\sqrt{\tr[ (\dev_3 (C\, C_p^{-1}))^2]}}\,\langle C_p^{-1}\,\widetilde{\Sigma}\, C_p, \dev_3 \widetilde{\Sigma}\rangle\notag\\&=
		-\frac{\lambda_{\rm p}^+}{\sqrt{\tr[ (\dev_3 (C\, C_p^{-1}))^2]}}\,\langle C_p^{-1}\,C, \dev_3 (C\, C_p^{-1})\rangle\notag\\
		&=
		-\frac{\lambda_{\rm p}^+}{\sqrt{\tr[ (\dev_3 (C\, C_p^{-1}))^2]}}\,\langle(C\, C_p^{-1})^T, \dev_3 (C\, C_p^{-1})\rangle\\
		&=
		-\frac{\lambda_{\rm p}^+}{\sqrt{\tr[ (\dev_3 (C\, C_p^{-1}))^2]}}\,\langle C^{-1/2}\dev_3(C\, C_p^{-1})\,C^{1/2}\, C^{-1/2}\,\dev_3 (C\, C_p^{-1})C^{1/2},\id\rangle\notag
		\\
		&=
		-\frac{\lambda_{\rm p}^+}{\sqrt{\tr[ (\dev_3 (C\, C_p^{-1}))^2]}}\,\langle \dev_3(C^{1/2}\, C_p^{-1}\, C^{1/2})^T\, \dev_3( C^{1/2}\,C_p^{-1}\, C^{1/2}),\id\rangle\notag
		\\
		&=
		-\frac{\lambda_{\rm p}^+}{\sqrt{\tr[ (\dev_3 (C\, C_p^{-1}))^2]}}\,\| \dev_3(C^{1/2}\, C_p^{-1}\, C^{1/2})\|^2,\notag
		\end{align}
		which is negative\footnote{Surprisingly, this follows even if  $C$ and $C_p^{-1}$ do not commute in general. If $C$ and $C_p$ commute, then  $X=C\, C_p^{-1}\in{\rm Sym}(3)$ and the quantity  does  have a sign, since then $\langle X^T, \dev_3 X\rangle=\|\dev_3 X\|^2\geq 0.$ }. Therefore, this   model is   thermodynamically correct as now  shown also for the simple Neo-Hooke energy.
		
		\subsection*{A.3 Another referential model}
		
		We recall that, in view of Lemma \ref{lemaplasticn},
		any isotropic free energy  $W$ defined in terms of $F_e$ can be expressed as
		$
		W(F_e)=\widetilde{W}(C\,C_p^{-1})
		$. In order to assume that the reduced dissipation inequality is satisfied, we compute
		\begin{align*}
		\frac{\rm d}{\rm dt} \widetilde{W}(C\, C_p^{-1})&=\langle D\widetilde{W}(C\,C_p^{-1}), C\, \frac{\rm d}{\rm dt} [C_p^{-1}]\rangle\notag\\&=\langle {C\,D\widetilde{W}(C\,C_p^{-1})\, C_p^{-1}},  \frac{\rm d}{\rm dt} [C_p^{-1}]\, C_p\rangle=\langle {\widetilde{\Sigma}},  \frac{\rm d}{\rm dt} [C_p^{-1}]\, C_p\rangle.
		\end{align*}
		Here, $
		\widetilde{\Sigma}=2\,C\, D_C[\widetilde{W}(C\, C_p^{-1})],
		$
		as in the Reese 2008 and Shutov-Ihlemann 2014 model.
		It is tempting to assume the flow rule in the associated form (see e.g. the habilitation thesis of Miehe \cite[page 73, Satz 5.32]{Miehe92} or \cite{miehe1998constitutive} and also \cite[Table 1]{miehe1995theory})
		\begin{align}\label{miehetildesigma}
		\frac{\rm d}{\rm dt} [C_p^{-1}]\, C_p\in-\,\Partial_{\widetilde{\Sigma}}{\Chi}(\dev_3\widetilde{\Sigma}),
		\end{align}
		where
		$ {\Chi}(\dev_3\widetilde{\Sigma})$ is the indicator function of the convex elastic domain
		\begin{align*}
		\mathcal{E}_{\rm e}(\widetilde{\Sigma},{\frac{2}{3}}\, \sigma_{\textbf{y}}^2):=\left\{\widetilde{\Sigma}\in \R^{3\times 3}\, | \, \|\dev_3\widetilde{\Sigma}\|^2\leq {\frac{2}{3}}\, \sigma_{\textbf{y}}^2\right\}.
		\end{align*}
		Note that this   flow rule \eqref{miehetildesigma}  is not  the formulation which Miehe seemed to intend. We have discussed the correct  interpretation  in Subsection \ref{Miehe1995}.
		
		Regarding such a formulation we can summarize our observations:
		\begin{itemize}
			\item[i)] this flow rule is thermodynamically correct;
			\item[ii)] the right hand side is  a function of $C$ and $C_p^{-1}$ only, i.e. $\widetilde{\Sigma}=\widetilde{\Sigma}(C,C_p^{-1})$;
			\item[iii)] plastic incompressibility: from this flow rule it follows that $\det C_p(t)=1$, since the right hand side is trace-free;
			\item[iv)] however, the computed tensor $\boldsymbol{C_p(t)}$ {\bf will  not be symmetric} since   $\widetilde{\Sigma}\,C_p^{-1}\not\in {\rm Sym}(3)$ in general. For instance, for  the simplest Neo-Hooke energy $W(F_e)=\tr(C_e)=\tr(C\, C_p^{-1})$ we have  $\widetilde{\Sigma}=2\,C\, C_p^{-1}\not\in {\rm Sym}(3)$, $\widetilde{\Sigma}\,C_p^{-1}=2\,C\, C_p^{-2}\not\in {\rm Sym}(3)$, in general, and the flow rule becomes
			\begin{align}
			\frac{\rm d}{\rm dt} [C_p^{-1}]\, =-2\,\frac{\lambda_{\rm p}^+}{\|\dev  (C\,C_p^{-1})\|}\, [C\,C_p^{-2}-\frac{1}{3}\, \tr(C\,C_p^{-1})\cdot C_p^{-1}]\not\in {\rm Sym}(3);
			\end{align}
			\item[v)] it is an associated plasticity model in the sense of Definition \ref{definitionpld}.
		\end{itemize}
		In conclusion, this model is inconsistent with the requirement for a plastic metric, i.e. $C_p\in {\rm Psym}(3)$. Moreover, if we are looking to the flow rule   in the associated form considered in the habilitation thesis of Miehe \cite[page 73, Satz 5.32]{Miehe92} (see \cite{miehe1998constitutive} and also \cite[Table 1]{miehe1995theory}),  since the subdifferential $ \Partial_{\widetilde{\Sigma}}{\Chi}(\dev_3\widetilde{\Sigma})$ of the indicator function $\Chi$ is the normal cone
		\begin{align*}
		\mathcal{N}(\mathcal{E}_{\rm e}(\widetilde{\Sigma},\frac{1}{3}\,{\boldsymbol{\sigma}}_{\!\mathbf{y}}^2);\dev_3 \widetilde{\Sigma})=\left\{\begin{array}{ll}
		0, & \widetilde{\Sigma}\in {\rm int}(\mathcal{E}_{\rm e}(\widetilde{\Sigma},\frac{1}{3}{\boldsymbol{\sigma}}_{\!\mathbf{y}}^2))
		\vspace{2mm}\\
		\{\lambda^+_{\rm p}\, \frac{\dev_3\widetilde{\Sigma}}{\|\dev_3\widetilde{\Sigma}\|}\,|\, \lambda^+_{\rm p}\in \R_+\},& \widetilde{\Sigma}\not\in {\rm int}(\mathcal{E}_{\rm e}(\widetilde{\Sigma},\frac{1}{3}{\boldsymbol{\sigma}}_{\!\mathbf{y}}^2)).
		\end{array}\right.
		\end{align*}
		the
		flow rule can be written in the form
		\begin{align}\label{mieheflowshu}
		\frac{\rm d}{\rm dt} [C_p^{-1}]\, C_p =-\lambda^+_{\rm p}\, \frac{\dev_3\widetilde{\Sigma}}{\|\dev_3\widetilde{\Sigma}\|}\,,
		\end{align}
		which is not equivalent with the flow rule \eqref{shutovremark} considered by Miehe in \cite{Miehe95}, since $\widetilde{\Sigma}\not\in{\rm Sym}(3)$. Let us remark that we have the symmetries $\dev_3\widetilde{\Sigma}\cdot C_p\in {\rm Sym}(3)$, $C_p^{-1}\, \dev_3 \widetilde{\Sigma}\in {\rm Sym}(3)$, but these do not assure that the flow rule \eqref{mieheflowshu}  implies $C_p\in {\rm Sym}(3)$.
		
		\subsection*{A.4. The Simo-Hughes 1998-model for the Saint-Venant-Kirchhoff energy and  for the Neo-Hooke energy}
		
		In order to see that the quantity $F_e^{-1}
		\frac{\dev_n\tau_e}{\|\dev_n\tau_e\|}F_e^{-T}$ which appears in the Simo-Hughes flow rule is not necessarily a trace free matrix,  we consider two energies: the isotropic elastic Saint-Venant-Kirchhoff energy and the energy considered by Simo and Hughes \cite[page 307]{Simo98b}.  On the one hand, the well known  isotropic elastic Saint-Venant-Kirchhoff energy is
		\begin{align*}
		W_{_{\rm SVK}} &=\frac{\mu}{4}\,\|C_e-\id\|^2+\frac{\lambda}{8}\,[\tr(C_e-\id)]^2
		=\frac{\mu}{4}\,\|B_e-\id\|^2+\frac{\lambda}{8}\,[\tr(B_e-\id)]^2,
		\end{align*}
		and the corresponding Kirchhoff stress tensor is given by
		\begin{align*}
		\tau_e^{\rm SVK}(U)&=D_{B_e}[W^{{\rm SVK}}(B_e)]=\mu\,(F_e^{-T}\, C_e\, F_e^T-\id)+\frac{\lambda}{2}\tr(C_e-\id)\cdot \id\notag\\&=\mu\,(F_e\, F_e^T-\id)+\frac{\lambda}{2}\tr(F_e\, F_e^T-\id)\cdot \id.\notag
		\end{align*}
		Hence, we deduce
		\begin{align*}
		F_e^{-1}
		\,[\dev_n\tau_e^{\rm SVK}]\,F_e^{-T}&=\mu \,F_e^{-1}
		\,{\dev_n[\,F_e\, F_e^T]}\,F_e^{-T}=\mu \,F_e^{-1}
		\,[\,F_e\, F_e^T-\frac{1}{3}\tr(F_e\, F_e^T)\cdot \id]\,F_e^{-T}\\
		&=\mu
		\,[\ \id-\frac{1}{3}\tr(F_e\, F_e^T)\cdot F_e^{-1}\,F_e^{-T}]\notag\\&=\mu
		\,[\ \id-\frac{1}{3}\tr(F_e\, F_e^T)\cdot F_e^{-1}\,F_e^{-T}],\notag
		\end{align*}
		and further
		\begin{align*}
		\langle F_e^{-1}
		\,[\dev_n\tau_e^{\rm SVK}]\,F_e^{-T},\id\rangle&=\mu
		\,\langle\ \id-\frac{1}{3}\tr(F_e\, F_e^T)\cdot F_e^{-1}\,F_e^{-T}],\id\rangle\notag\\&=\mu\left[3-
		\,\frac{1}{3}\tr(F_e\, F_e^T)\,\tr( F_e^{-1}\,F_e^{-T})\right]\\\notag
		&=\mu\left[3-
		\,\frac{1}{3}\tr(C_e)\,\tr( C_e^{-1})\right]\notag\\&=\mu\left[3-
		\,\frac{1}{3\, \det C_e}\tr(C_e)\,\tr({\rm Cof} \,C_e)\right].
		\end{align*}
		We remark that $\langle F_e^{-1}
		\,[\dev_n\tau_e^{\rm SVK}]\,F_e^{-T},\id\rangle=0$ if and only if
		$
		\tr(C_e)\,\tr({\rm Cof} \,C_e)=9\, \det C_e,
		$ which does not hold true in general. Since $C_e $ and ${\rm Cof} C_e$ are coaxial and symmetric, the problem can be reduced to the diagonal case, i.e.  we may assume $C_e=\diag (\lambda_1,\lambda_2,\lambda_3)$, $\lambda_i>0$. Hence the condition $
		\tr(C_e)\,\tr({\rm Cof} \,C_e)=9\, \det C_e,
		$ becomes
		\begin{align*}
		9\,\lambda_1\lambda_2\lambda_3&=(\lambda_1+\lambda_2+\lambda_3)(\lambda_1\lambda_2+\lambda_2\lambda_3+\lambda_3\lambda_2) \notag\\&\Leftrightarrow \quad 0=\lambda_1(\lambda_2-\lambda_3)^2+\lambda_2(\lambda_3-\lambda_1)^2+\lambda_3(\lambda_1-\lambda_3)^2\notag
		\end{align*}
		which is satisfied if and only if $\lambda_1=\lambda_2=\lambda_3$. Therefore, for the Saint-Venant-Kirchhoff energy, in this model, $\det \overline{C}_p=1$ is only true for the conformal mapping $F_e=\lambda\cdot {\rm SO}(3)\in \R_+\cdot {\rm SO}(3)$.
		
		On the other hand, the energy considered by Simo and Hughes \cite[page 307]{Simo98b} is
		\begin{align*}
		W_{\rm Simo}(B_e)=\frac{\mu}{2}\langle \frac{B_e}{\det B_e^{1/3}}-\id,\id\rangle+\frac{\kappa}{4}\left[(\det B_e-1)-\log (\det B_e)\right],
		\end{align*}
		for which the Kirchhoff stress tensor is given by
		\begin{align*}
		\tau_e^{{\rm Simo}}=\mu \dev_3\left( \frac{B_e}{\det B_e^{1/3}}\right)+\frac{\kappa}{2}\, \left(J_e -\frac{1}{J_e}\right)\cdot \id.
		\end{align*}
		Hence, we deduce
		\begin{align*}
		\langle& F_e^{-1}\,[{\dev_n\tau_e^{{\rm Simo}}}]\,F_e^{-T},\id \rangle=\mu \, \frac{1}{\det B_e^{1/3}}\,\langle F_e^{-1}[\dev_3 {B_e}]\,F_e^{-T},\id\rangle\notag\\&=\mu \, \frac{1}{\det B_e^{1/3}}\,\langle \dev_3 {B_e},F_e^{-T}\,F_e^{-1}\rangle
		\\&=\mu \, \frac{1}{\det B_e^{1/3}}\,\langle \dev_3 {B_e},B_e^{-1}\rangle=\mu \, \frac{1}{\det B_e^{1/3}}\left[\,\langle {B_e},B_e^{-1}\rangle-\frac{1}{3}\tr( {B_e})\,\tr(B_e^{-1})\right]\notag\\
		&=\mu \, \frac{1}{\det B_e^{1/3}}\left[3-\frac{1}{3}\tr( {B_e})\,\tr(B_e^{-1})\right]\notag\\&
		=\mu \, \frac{1}{\det B_e^{4/3}}\left[3\,\det B_e-\frac{1}{3}\tr( {B_e})\,\tr({\rm Cof}\,B_e)\right].\notag
		\end{align*}
		Therefore $\langle F_e^{-1}
		\,[\dev_n\tau_e^{{\rm Simo}}]\,F_e^{-T},\id\rangle=0$ if and only if
		$
		9\,\det B_e=\tr( {B_e})\,\tr({\rm Cof}\,B_e)
		$. Similar as above, it follows that  this holds true  if and only if  $F_e=\lambda\cdot {\rm SO}(3)\in \R_+\cdot {\rm SO}(3)$.
	
\end{document}